\documentclass[unicode,runningheads,envcountsame,envcountsect,a4paper,dvipsnames]{llncs} % envcountsect を削除
\usepackage[T1]{fontenc}

\usepackage{orcidlink}
\renewcommand{\orcidID}[1]{\orcidlink{#1}}

\usepackage{subcaption}
\usepackage{amsmath,amssymb,mathtools}

\usepackage[capitalize,nameinlink]{cleveref}
\hypersetup{hidelinks}

\spnewtheorem{myclaim}{Claim}[theorem]{\itshape}{\upshape}
\crefname{myclaim}{Claim}{Claims}
\spnewtheorem{observation}[theorem]{Observation}{\bfseries}{\upshape}
\crefname{observation}{Observation}{Observations}

\spnewtheorem*{proofsketch}{Proof sketch}{\itshape}{\normalfont}

\usepackage{thm-restate}

\usepackage{xcolor}
\usepackage{enumitem}
\setlist[itemize]{itemsep=3pt,parsep=0pt,topsep=2pt}
\usepackage[ruled,linesnumbered,vlined]{algorithm2e}
\usepackage{tcolorbox}
\usepackage{booktabs}

\usepackage{complexity}
\usepackage{bbding}

\usepackage{xcolor}
\definecolor{darkred}{rgb}{0.7,0,0}

\usepackage{CJKutf8}

\newcommand{\cutset}{\Gamma}

\newcommand{\bigOh}{\mathcal{O}}
\newcommand{\bigoh}{\mathcal{O}}

\newcommand{\weight}{\omega}
\newcommand{\symdif}{\mathbin{\triangle}}

\newcommand{\fgmincut}{\textsc{Forcing Global Min Cut}}
\newcommand{\afgmincut}{\textsc{Anti-Forcing Global Min Cut}}
\newcommand{\fminstcut}{\textsc{Forcing Min $s$--$t$ Cut}}
\newcommand{\afminstcut}{\textsc{Anti-Forcing Min $s$--$t$ Cut}}
\newcommand{\hgmc}{\textsc{Hitting Global Min Cut}}
\newcommand{\uta}{\textsc{Unweighted Tree Augmentation}}

\let\doendproof\endproof
\renewcommand\endproof{~\hfill$\qed$\doendproof}

\newenvironment{claimproof}{%
  \renewcommand\endproof{~\hfill$\lozenge$\doendproof}
  \begin{proof}%
}{%
  \end{proof}%
}

\title{On the Complexity of Forcing and Anti-Forcing Minimum Cuts\thanks{%
Partially supported
by JSPS KAKENHI Grant Numbers 
JP22H00513, % Yota: Ono Kiban A
JP23K28034, % Yasuaki: Kobayashi Kiban B
JP24H00686, % Yasuaki: Ito kiban A
JP24H00697, % Yasuaki, Yota: Tamaki Kiban A
JP24K23847, % Tatsuya: Gima Startup
JP25K03076, % Yota: Otachi Kiban B
JP25K03077, % Yota, Tatsuya: Hanaka Kiban B
JP26K02980, % Tatsuya, Yasuaki: Arimura Kiban B
JP26K21161. % Tatsuya: Gima Wakate
}}

\author{
Tatsuya Gima\inst{1}\orcidID{0000-0003-2815-5699} \and
Yasuaki Kobayashi\inst{1}\orcidID{0000-0003-3244-6915} \and
Hiraku Morimoto\inst{2}\Envelope\orcidID{0009-0002-4646-4102}\and \\
Yota Otachi\inst{3}\orcidID{0000-0002-0087-853X}
}

\authorrunning{T. Gima et al.}
\institute{%
Hokkaido University, Sapporo, Japan.
\email{gima@ist.hokudai.ac.jp}, 
\email{koba@ist.hokudai.ac.jp} 
\and
Nagoya University, Nagoya, Japan.
\email{hiraku.morimoto@nagoya-u.jp}
\and
Institute of Science Tokyo, Tokyo, Japan.
\email{otachi@comp.isct.ac.jp}
}

\begin{document}

\maketitle

\begin{abstract}
For an instance of a combinatorial optimization problem, a \emph{forcing set} is a set of elements such that there is a unique optimal solution including it. Symmetrically, an \emph{anti-forcing set} is a set of elements such that there is a unique optimal solution excluding it. In this paper, we study the problems of computing smallest forcing and anti-forcing sets for two classical cut problems, \textsc{Global Min Cut} and \textsc{Min $s$--$t$ Cut}. We also consider variants in which the optimal cut to be uniquely determined is given as input. For each of these problems, we either give a polynomial-time algorithm or prove \NP-completeness.

% keywords are separated by \and commands
\end{abstract}

%%%%%%%%%%%%%%%%%%%%%%%%%%%%%%%%%%%%%%%%%%%%%%%%%%%%%%%%%%%%%%%%%%%%%%%%%%%%%%%%%%%%%%%%%%%%%%%%%%%%%%%%%%%%%%%%%%%%%%%%%%%%%%%%
\section{Introduction}

For a combinatorial optimization problem, let $\mathcal X^* \subseteq 2^E$ be the set of all optimal solutions on an input with ground set $E$.
A \emph{forcing set} for $\mathcal X^*$ is a subset $F \subseteq E$ such that there is exactly one solution $X \in \mathcal X^*$ that includes it, that is, $F \subseteq X$ and $F \setminus Y \neq \emptyset$ for $Y \in \mathcal X^*$ with $Y \neq X$.
An \emph{anti-forcing set} for $\mathcal X^*$ is a subset $F \subseteq E$ such that there is exactly one solution $X \in \mathcal X^*$ that excludes it, that is, $F \subseteq E \setminus X$ and $F \cap Y \neq \emptyset$ for $Y \in \mathcal X^*$ with $Y \neq X$.
In other words, for a forcing set, there is a unique optimal solution including it, and for an anti-forcing set, there is a unique optimal solution excluding it.

Forcing sets and their ``dual'' concept, anti-forcing sets, were introduced by Harary, Klein, and \v{Z}ivkovi\'{c} \cite{HararyKZ91} and Vuki\v{c}evi\'{c} and Trinajsti\'{c}~\cite{vukicevic2007anti}, respectively.
In the original setting, $\mathcal X^*$ is considered the collection of all perfect matchings of a graph, and due to their important applications in computational chemistry, they have been investigated from graph-theoretic and computational perspectives~\cite{AdamsMM04,AfshaniHM04,aoshima2026hardness,HararyKZ91,vukicevic2007anti} (see \cite{ZhangHLZ25} for a survey).
Although the initial work primarily focused on perfect matchings, the concept can be naturally extended to various combinatorial (optimization) problems~\cite{AnCCKLOS26,DemaineMSWA16,gima2025forcing,HararySV07,HatamiM05,HoriyamaKOSS24,HoriyamaSSS25}.

The computational complexity of the problem of finding such a smallest ``fingerprint'' has been investigated under different names, such as defining sets~\cite{HatamiM05}, fewest clues~\cite{DemaineMSWA16}, and pre-assignments~\cite{AnCCKLOS26,HoriyamaKOSS24,HoriyamaSSS25}.
For perfect matchings, the authors of \cite{AdamsMM04,AfshaniHM04,aoshima2026hardness} showed that the decision versions of the problems of finding smallest forcing and anti-forcing sets are both \NP-complete even for restricted classes of graphs, while the problem of computing a perfect matching of a graph is solvable in polynomial time.
This phenomenon, in which the problem of finding a smallest fingerprint is harder than that of finding an optimal solution, has been observed in various problems, e.g., the decision problems of finding smallest forcing and anti-forcing sets for \textsc{Vertex Cover} are $\mathrm{\Sigma}^{\P}_2$-complete~\cite{HoriyamaKOSS24}.
Most closely related to our work, the authors of \cite{gima2025forcing} investigated the (parameterized) complexity of computing smallest forcing and anti-forcing sets for \textsc{Shortest $s$--$t$ Path} and \textsc{Minimum Spanning Tree}.
They showed that the decision version of the problem of finding a smallest anti-forcing set for \textsc{Shortest $s$--$t$ Path} is \NP-complete, whereas the other three problems are polynomial-time solvable. 
Thus, unlike the aforementioned examples, three of these four problems remain polynomial-time solvable.

Motivated by the results of \cite{gima2025forcing}, in this paper, we consider the problem of computing smallest forcing and anti-forcing sets for classical cut problems, \textsc{Global Min Cut} and \textsc{Min $s$--$t$ Cut}.
Let $G = (V, E, \weight)$ be an edge-weighted graph with a weight function $\weight \colon E \to \mathbb Z_{>0}$.
Let $\Lambda(G)$ (resp.~$\Lambda_{st}(G)$) be the set of all cutsets of minimum cuts (resp.~minimum $s$--$t$ cuts) of $G$.
A \emph{forcing set} for $\Lambda(G)$ (resp.~$\Lambda_{st}(G)$) is an edge subset $F \subseteq E$ such that there is a unique cutset $\cutset_{G}(X)$ in $\Lambda(G)$ (resp.~in $\Lambda_{st}(G)$) such that $F \subseteq \cutset_G(X)$.
An \emph{anti-forcing set} for $\Lambda(G)$ (resp.~$\Lambda_{st}(G)$) is an edge subset $F \subseteq E$ such that there is a unique cutset $\cutset_{G}(X)$ in $\Lambda(G)$ (resp.~in $\Lambda_{st}(G)$) such that $\cutset_G(X) \subseteq E(G) \setminus F$.
In this paper, we consider the following problems.

\begin{tcolorbox}
\begin{description}
    \item[Problem] \fgmincut{} / \afgmincut{}
    \item[Input: ] An edge-weighted graph $G = (V, E, \weight)$ and an integer $k$.
    \item[Goal: ] Determine if $G$ has a forcing set / an anti-forcing set $F \subseteq E$ for $\Lambda(G)$ with size at most $k$.
\end{description}
\end{tcolorbox}

\begin{tcolorbox}
\begin{description}
    \item[Problem] \fminstcut{} / \afminstcut{}
    \item[Input: ] An edge-weighted graph $G = (V, E, \weight)$, distinct vertices $s, t \in V$, and an integer $k$.
    \item[Goal: ] Determine if $G$ has a forcing set / an anti-forcing set $F \subseteq E$ for $\Lambda_{st}(G)$ with size at most $k$.
\end{description}
\end{tcolorbox}

We also consider natural variants of these problems, where we are additionally given a specific minimum cut (resp.~minimum $s$--$t$ cut) $\{X, V \setminus X\}$ and asked to determine whether $G$ has a forcing set $F \subseteq E$ for $\Lambda(G)$ (resp.~$\Lambda_{st}(G)$) with size at most $k$ such that $F \subseteq \cutset_G(X)$.
We refer to these problems as \fgmincut{} and \fminstcut{} \emph{with a specified cut}, respectively.
Similarly, \afgmincut{} and \afminstcut{} \emph{with a specified cut} are defined as well.

\paragraph{Our results.} We prove that \fgmincut{} is polynomial-time solvable and \afgmincut{} is \NP-complete.
In particular, by exploiting cactus representations of minimum cuts~\cite{DinitzKL76}, we show that every global minimum cut admits a forcing set of size at most 2.
Based on this structural result, we give a near-linear-time algorithm for \fgmincut{}.
These positive and negative results remain unchanged when a target minimum cut, which is to be unique, is given as input.
For minimum $s$--$t$ cuts, a tractability boundary can be drawn ``orthogonally'' to the previous cases: \fminstcut{} and \afminstcut{} are both \NP-complete, while they can be solved in polynomial time when a target minimum $s$--$t$ cut is given as input.
These results are shown by means of a representation theorem for minimum $s$--$t$ cuts due to Picard and Queyranne~\cite{PicardQ80}.
Our results are summarized in \cref{tab:results}.
\begin{table}[t]
    \centering
    \caption{A summary of our results.}\label{tab:results}
    \begin{minipage}[t]{0.47\linewidth}
        \centering
        \caption*{\textsc{Global Min Cut}}
        \begin{tabular}{lcc}
            \toprule
             & \ Forcing \ & Anti-forcing \\
            \midrule
            Cut not specified & \cP & \NP-c. \\
            Specified cut     & \cP & \NP-c. \\
            \bottomrule
        \end{tabular}
    \end{minipage}
    \hfill
    \begin{minipage}[t]{0.47\linewidth}
        \centering
        \caption*{\textsc{Min $s$--$t$ Cut}}
        \begin{tabular}{lcc}
            \toprule
             & \ Forcing \ & Anti-forcing \\
            \midrule
            Cut not specified & \NP-c. & \NP-c. \\
            Specified cut     & \cP & \cP \\
            \bottomrule
        \end{tabular}
    \end{minipage}
\end{table}
% Proofs of the starred results appear in the full version~\cite{GimaKMO26full}.
Proofs of the starred results appear in the appendix.

%%%%%%%%%%%%%%%%%%%%%%%%%%%%%%%%%%%%%%%%%%%%%%%%%%%%%%%%%%%%%%%%%%%%%%%%%%%%%%%%%%%%%%%%%%%%%%%%%%%%%%%%%%%%%%%%%%%%%%%%%%%%%%%%

\section{Preliminaries}

Let $G$ be a (directed) graph.
The vertex set and edge set of $G$ are denoted by $V(G)$ and $E(G)$, respectively.
For $X \subseteq V(G)$, we denote by $G[X]$ the subgraph of $G$ induced by $X$.
For disjoint vertex sets $X$ and $Y$, the set of edges between $X$ and $Y$ is denoted by $E_G(X, Y)$.
For a weight function $\weight$ on edges, we define $\weight(X, Y) = \sum_{e \in E_G(X, Y)} \weight(e)$.
For an undirected graph, we may write $uv$ to denote an edge between two vertices $u$ and $v$.
For $F \subseteq E(G)$, the graph obtained from $G$ by deleting all edges in $F$ is denoted by $G - F$.

A \emph{cut} of $G$ is a bipartition $\{X,V(G)\setminus X\}$ into two
nonempty sets, and its \emph{cutset} is $\Gamma_G(X) \coloneqq E_G(X,V(G)\setminus X)$.
The sets $X$ and $V(G)\setminus X$ are called the \emph{shores} of the cut. For an $s$--$t$ cut, the shore containing $s$ is called the \emph{$s$-shore}.
The \emph{weight} of a cut $\{X, V\setminus X\}$ is defined as $\weight(X, V(G) \setminus X)$.
We denote by $\lambda(G)$ and by $\lambda_{st}(G)$ the minimum weight of a cut and an $s$--$t$ cut of $G$, respectively.
Define
\begin{align*}
    \Lambda(G) &\coloneqq \{\cutset_{G}(X) : \emptyset \neq X \subsetneq V(G),\, \weight(X, V(G) \setminus X) = \lambda(G)\},\\
    \Lambda_{st}(G) &\coloneqq \{\cutset_{G}(X) : s \in X,\, t \notin X,\, \weight(X, V(G)\setminus X) = \lambda_{st}(G)\}. 
\end{align*}

While the cutset corresponding to a cut $\{X, V(G) \setminus X\}$ is uniquely determined as $\cutset_G(X)$, from a cutset $S \subseteq E(G)$, the corresponding cut $\{X, V(G) \setminus X\}$ with $S = \cutset_G(X)$ may not be uniquely determined.
However, when $G$ is connected, $\weight$ is positive, and $S$ is the cutset of a minimum ($s$--$t$) cut of $G$, the corresponding cut is uniquely determined as $G - S$ has exactly two connected components.
Moreover, when $G$ is disconnected, our problems are trivial or reduced to the case where $G$ is connected.
Thus, throughout this paper, we may assume that $G$ is connected and $\weight$ is positive.

We next describe a transformation for removing edge weights and parallel edges, which will be useful in our later reductions.
% The proof of the following lemma is given in the appendix.
\begin{restatable}[$\star$]{lemma}{transformUnweightedSimple}
    \label{lem:transform-unweighted-simple}
    Let $G$ be an edge-weighted multigraph.
    Then $G$ can be transformed into an unweighted simple graph $G'$ such that the minimum cuts of $G$ and $G'$ are in one-to-one correspondence.
    The same holds for minimum $s$--$t$ cuts.
\end{restatable}
The transformation preserves the minimum sizes of hitting, forcing, and
anti-forcing sets and runs in polynomial time when the edge weights are
polynomially bounded.
% \begin{proof}
%     First, replace each edge $e$ with $\weight(e)$ parallel unweighted edges.
%     This preserves the value and the shores of every cut.

%     Let $H$ be the resulting unweighted multigraph and let $m=|E(H)|$.
%     Replace each vertex $v$ by a clique $K_v$ of size $m+2$, and represent each edge $uv$ of $H$ by an edge between distinct vertices of $K_u$ and $K_v$.
%     Any cut splitting a clique has size at least $m+1$, whereas every cut corresponding to a cut of $H$ has size at most $m$.
%     Hence, no minimum cut splits a clique, and the claimed correspondence follows.

%     For minimum $s$--$t$ cuts, apply the same construction and choose one vertex of $K_s$ and one vertex of $K_t$ as the new terminals $s$ and $t$, which proves the latter claim.
% \end{proof}
% Note that if all edge weights are polynomially bounded, then the transformation can be performed in polynomial time.
% These transformations also preserve the minimum cardinalities of forcing and anti-forcing sets.
% Indeed, in the first transformation, each original edge in such a set can be replaced by an arbitrary edge in the corresponding parallel bundle, and in the second transformation by its corresponding edge between the cliques.
% Conversely, the transformation can be reversed by mapping these edges back to the original edges and discarding clique edges, which belong to no minimum cut.

%%%%%%%%%%%%%%%%%%%%%%%%%%%%%%%%%%%%%%%%%%%%%%%%%%%%%%%%%%%%%%%%%%%%%%%%%%%%%%%%%%%%%%%%%%%%%%%%%%%%%%%%%%%%%%%%%%%%%%%%%%%%%%%%

\section{Global Min Cut}

We start this section with a well-known \emph{cactus representation} of minimum cuts, which is a key technical ingredient for our near-linear-time algorithm for \fgmincut{}.

\subsection{Representing all minimum cuts}

We call a multigraph $\mathcal{K}$ a \emph{cactus} if every edge of $\mathcal{K}$ belongs to exactly one cycle.\footnote{Although the standard definition of cacti allows bridges, we forbid them here and allow cycles of length~$2$ instead.}
%We call a multigraph $\mathcal{K}$ a \emph{cactus} if it is $2$-edge-connected and every edge of $\mathcal{K}$ belongs to exactly one cycle\footnote{Although a cactus can have bridges in general, it can be assumed to be $2$-edge-connected by replacing each bridge with a pair of parallel edges.}.
% Here, cycles of length $2$, consisting of a pair of parallel edges, are allowed.
Let $G$ be a connected edge-weighted graph, let $\mathcal{K}$ be a connected unweighted cactus, and let $\varphi \colon V(G)\to V(\mathcal{K})$ be a mapping.
We say that $(\mathcal{K},\varphi)$ is a \emph{cactus representation} of $G$ if the following conditions hold:
\begin{enumerate}
  \item for every $\cutset_G(X) \in\Lambda(G)$, there exists a minimum cut
        $\{Y, V(\mathcal{K})\setminus Y\}$ of $\mathcal{K}$ such that $\{X, V(G)\setminus X\}=\{\varphi^{-1}(Y),V(G)\setminus\varphi^{-1}(Y)\}$; and
  \item for every minimum cut $\{Y, V(\mathcal{K})\setminus Y\}$ of $\mathcal{K}$, we have
        $\cutset_G(\varphi^{-1}(Y))\in\Lambda(G)$.
\end{enumerate}
We emphasize that $G$ is edge-weighted, whereas the cactus $\mathcal{K}$ is unweighted.
See \cref{fig:cactus-representation} for an example of cactus representation.
\begin{figure}[t]
  \centering
  \includegraphics[scale=0.8]{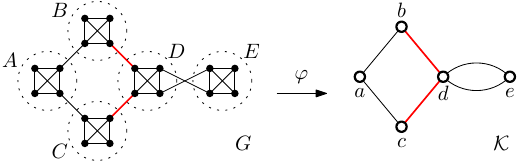}
  \caption{A graph $G$ and a cactus representation $(\mathcal{K}, \varphi)$ of $G$,
  where $\varphi$ maps $A, B, C, D, E$ to $a, b, c, d, e$, respectively.
  The two red edges of $\mathcal{K}$ represent the minimum cut $\{A\cup B\cup C, D\cup E\}$ of $G$.}
  \label{fig:cactus-representation}
\end{figure}
% When $(\mathcal{K}, \varphi)$ is a cactus representation of $\Lambda(G)$, we simply call it a cactus representation of $G$.
A minimum cut of a cactus is obtained by deleting two edges $f_1$ and $f_2$ on a common cycle and taking the vertex set of either resulting component.
We say that $f_1$ and $f_2$ \emph{represent} this cut (or its cutset).
A vertex $v$ in $\mathcal{K}$ is \emph{empty} if $\varphi^{-1}(v)$ is empty and 
a \emph{$k$-junction} if $v$ is contained in exactly $k$ cycles of $\mathcal{K}$.
Dinitz et al.~\cite{DinitzKL76} proved that every connected edge-weighted graph $G$ admits a cactus representation.
% Moreover, such a representation can be constructed in $\bigoh(m\,\mathrm{polylog}\,n)$ time~\cite{HeHS24a}, where $n$ and $m$ are the numbers of vertices and edges in $G$.
Moreover, such a representation can be constructed deterministically in $\bigoh(m\,\mathrm{polylog}\,n)$ time, up to a logarithmic factor in the total edge weight~\cite{HeHS24a,HeHS24full,HenzingerLRW24}.
Here $n=|V(G)|$ and $m=|E(G)|$.

\subsection{\fgmincut{} in near-linear time}
First, we show that \fgmincut{} can be solved in near-linear time, up to a logarithmic factor in the total edge weight.
Before presenting the algorithm, we establish several properties of cactus representations that will be useful in the proof.
Let $G=(V,E,\weight)$ be an edge-weighted graph, and let $(\mathcal{K},\varphi)$ be a cactus representation of $G$.
We denote $\lambda:=\lambda(G)$.
We only consider the case $\lambda>0$, that is, $G$ is connected.
For a vertex $v$ on a cycle $C$ of $\mathcal{K}$, let $e$ and $f$ be the two edges of $C$ incident with $v$.
We denote by $\mathcal{K}[C,v]$ the connected component of $\mathcal{K}-\{e,f\}$ that contains $v$.

Lo, Schmidt, and Thorup~\cite{LoST21} stated the following lemma for unweighted graphs, but their proof works for edge-weighted graphs simply by replacing edge-set sizes with their total weights.
\begin{lemma}[{\cite[Lemma~5]{LoST21}}]
  \label{lem:fgm:half}
  Let $u$ and $v$ be two distinct vertices on a cycle $C$ of length at least $3$ in $\mathcal{K}$.
  If $u$ and $v$ are adjacent on $C$, then the total weight of the edges in $G$ between
  $\varphi^{-1}(\mathcal{K}[C,u])$ and $\varphi^{-1}(\mathcal{K}[C,v])$ is exactly $\lambda/2$.
  On the other hand, if $u$ and $v$ are not adjacent on $C$, then there is no edge in $G$ between
  $\varphi^{-1}(\mathcal{K}[C,u])$ and $\varphi^{-1}(\mathcal{K}[C,v])$.
\end{lemma}
% \begin{proof}
%   Suppose $u$ and $v$ are adjacent in $C$.
%   Let $X_1 := \varphi^{-1}(\mathcal{K}[C, u])$, $X_2 := \varphi^{-1}(\mathcal{K}[C, v])$, and $X_3 := V \setminus (X_1 \cup X_2)$.
%   For each $i\in \{1, 2, 3\}$, the set $X_i$ is the inverse image under $\varphi$ of a shore of a minimum cut of $\mathcal{K}$.
%   Hence, $\cutset_G(X_i)\in \Lambda(G)$.
%   This implies that $\weight(X_1, X_2)+\weight(X_1, X_3) = \weight(X_2, X_3)+\weight(X_2, X_1) = \weight(X_3, X_1)+\weight(X_3, X_2) = \lambda$.
%   Hence, $\weight(X_1, X_2) = \lambda/2$.

%   If $u$ and $v$ are not adjacent in $C$, then there exist vertices $w_1, w_2\neq v$ adjacent to $u$ in $C$.
%   Let $X_1 := \varphi^{-1}(\mathcal{K}[C, u])$, $X_2 := \varphi^{-1}(\mathcal{K}[C, v])$, and $X_{i+2} := \varphi^{-1}(\mathcal{K}[C, w_i])$ for $i=1, 2$.
%   As $w_i$ is adjacent to $u$, we have that $\weight(X_1, X_3) = \weight(X_1, X_4) = \lambda/2$.
%   Since $\lambda = \weight(X_1)\ge \sum_{i=2, 3, 4}\weight(X_1, X_i) = \lambda+\weight(X_1, X_2)$, we have $\weight(X_1, X_2) = 0$.
% \end{proof}

\begin{lemma}
  \label{lem:fgm:path}
  Let $uv \in E(G)$, and let $P$ be a shortest $\varphi(u)$--$\varphi(v)$ path in $\mathcal{K}$. 
  Let $C$ be a cycle of $\mathcal{K}$ with length at least $3$.
  If there is an edge $xy \in E(C)$ such that $\varphi(u) \in \mathcal{K}[C,x]$ and $\varphi(v) \in \mathcal{K}[C,y]$, then $E(P) \cap E(C) = \{xy\}$.
  Moreover, if the cutset of a minimum cut of $G$ containing $uv$ is represented by two edges of a cycle of $\mathcal{K}$, then one of these two edges belongs to $P$.
\end{lemma}
\begin{proof}
  Let $C$ be a cycle of $\mathcal{K}$ of length at least $3$, and let $xy \in E(C)$ satisfy $\varphi(u) \in \mathcal{K}[C,x]$ and $\varphi(v) \in \mathcal{K}[C,y]$. 
  Since $\mathcal{K}$ is a cactus, $E(P)\cap E(C)$ is the $x$--$y$ subpath of $P$. 
  As $P$ is shortest and $xy\in E(C)$, this $x$--$y$ subpath is the single edge $xy$. 
  Thus, $E(P)\cap E(C)=\{xy\}$.
  
  Let $f_1$ and $f_2$ be two edges of a cycle $C$ that represent a minimum cut containing $uv$.
  Then $\varphi(u)$ and $\varphi(v)$ lie in different components of $\mathcal{K}-\{f_1,f_2\}$.
  Therefore, every $\varphi(u)$--$\varphi(v)$ path, and in particular $P$, contains at least one of $f_1$ and $f_2$.
\end{proof}

The following lemma allows us to reduce redundant vertices in a cactus representation.

\begin{restatable}[$\star$]{lemma}{forcingGlobalReduce}
  \label{lem:fgm:reduce}
  Let $C$ be a $2$-cycle with $V(C)=\{z,u\}$ and let $z$ be an empty $2$-junction.
  Then identifying $z$ and $u$ and deleting the two edges of $C$ yields a cactus representation of $G$.
\end{restatable}
% \begin{proof}
%   Let $\pi\colon V(\mathcal{K})\to V(\mathcal{K}')$ be the natural projection onto the cactus $\mathcal{K}'$ obtained by the stated identification, and put $\varphi':=\pi\circ\varphi$.
%   Let $C'$ be the other cycle containing $z$, and let $e_1,e_2$ be the two edges of $C'$ incident with $z$.
%   The cut represented by the two edges of $C$ and the cut represented by $\{e_1,e_2\}$ differ only in the placement of $z$.
%   Since $\varphi^{-1}(z)=\emptyset$, they induce the same cut of $G$.
%   Therefore, the cut represented by $C$ is redundant.
%   For every cut represented by a cycle different from $C$, the vertices $z$ and $u$ lie on the same side, and hence identifying them does not change its inverse image.
%   Thus, $(\mathcal{K}',\varphi')$ is a cactus representation of $G$.
% \end{proof}
A cactus representation $(\mathcal{K},\varphi)$ is \emph{reduced}
if $\mathcal{K}$ contains no $2$-cycle having an empty $2$-junction as an endpoint.
By \cref{lem:fgm:reduce}, every cactus representation can
be transformed into an equivalent reduced one.
In what follows, we assume that all cactus representations are reduced.

Fix a cycle $C$ in $\mathcal{K}$.
For $v \in V(C)$, let $A_v \coloneqq \varphi^{-1}(\mathcal{K}[C,v])$.
Note that $A_v$ is nonempty for every $v \in V(C)$.
This can be seen as follows: the two edges of $C$ incident with $v$ represent a minimum cut of $G$ having $A_v$ as one of its shores, which is nonempty as $\lambda > 0$.
For each edge $e=\{u,v\}\in E(G)$, fix a shortest $\varphi(u)$--$\varphi(v)$ path $P_e$ in $\mathcal{K}$.
\begin{restatable}[$\star$]{lemma}{forcingGlobalTwo}
  \label{lem:fgm:two}
  Every minimum cut of $G$ has a forcing set of size at most~$2$.
\end{restatable}
\begin{proofsketch}
  Let $S\in\Lambda(G)$, and let $e_1, e_2$ be two edges of a cycle $C$ in $\mathcal{K}$ representing the cutset $S$.
  We distinguish three cases.
  
  \medskip
  \noindent
  \emph{Case 1: $C$ has at least three edges, and $e_1,e_2$
  have no common endpoint.}
  
  Write $e_1=xy$ and $e_2=wz$.
  By \cref{lem:fgm:half}, we have $\weight(A_x,A_y)=\lambda/2>0$ and $\weight(A_z,A_w)=\lambda/2>0$.
  Hence, there exist edges $f_1\in E_G(A_x,A_y)$ and $f_2\in E_G(A_z,A_w)$ that belong to $S$.
  We claim that $\{f_1,f_2\}$ is a forcing set for $S$.
 
  Suppose for contradiction that there exists $S'\in\Lambda(G)\setminus\{S\}$ such that $\{f_1,f_2\}\subseteq S'$.
  Let $e'_1,e'_2$ be the two edges of a cycle $C'$ representing $S'$.
  By \cref{lem:fgm:path}, for each $i\in\{1,2\}$, we have $E(P_{f_i})\cap E(C)=\{e_i\}$ and $\{e'_1,e'_2\}\cap E(P_{f_i})\neq\emptyset$.
  If $C'=C$, it follows that $e_i\in\{e'_1,e'_2\}$ for each $i\in\{1,2\}$.
  Hence, $\{e'_1,e'_2\}=\{e_1,e_2\}$, which implies $S'=S$, a contradiction.
  Thus, $C'\neq C$.
  Since $e_1$ and $e_2$ have no common endpoint, no component of $\mathcal K-E(C)$ contains edges from both $P_{f_1}$ and $P_{f_2}$.
  Since $\mathcal K$ is a cactus and $C'\neq C$, the cycle $C'$ lies in one such component.
  However, since $\{e'_1,e'_2\}\cap E(P_{f_i})\neq\emptyset$ for each $i\in\{1,2\}$, the cycle $C'$ contains an edge from each of $P_{f_1}$ and $P_{f_2}$, a contradiction.

  \medskip
  \noindent
  \emph{Case 2: $C$ has at least three edges, and $e_1,e_2$
  share a common endpoint.}
  
  % The proof of this case is given in the full version~\cite{GimaKMO26full}.
  The proof of this case is given in the appendix.

  \medskip
  \noindent
  \emph{Case 3: $C$ is a $2$-cycle.}
  
  Let $x$ and $y$ be the endpoints of $e_1$ and $e_2$.
  Then $S=E_G(A_x,A_y)$.
  If $|S|=1$, its unique edge is a forcing set.
  Thus, we assume that $|S|\ge 2$.

  For $f \in S$, let $x(f)$ and $y(f)$ be the endpoints that are contained in $A_x$ and $A_y$, respectively. 
  We say that two edges $f_1$ and $f_2$ of $S$ are \emph{$x$-separated} if $\varphi(x(f_1))$ and $\varphi(x(f_2))$ lie in distinct components of $\mathcal{K}[C,x]-x$, or if at least one of them is $x$.
  We can define \emph{$y$-separated} edges analogously.
  Let $f_1, f_2\in S$ be both $x$-separated and $y$-separated.
  Suppose that $\{f_1, f_2\}$ is not a forcing set for $S$.
  Then some $S'\in\Lambda(G)\setminus \{S\}$ contains $f_1, f_2$ and is represented by two edges of a cycle $C'$.
  By \cref{lem:fgm:path}, each of $P_{f_1}$ and $P_{f_2}$ contains at least one of these two edges.
  Since $C$ has length $2$ and $S\neq S'$, we have $C'\neq C$.
  As $\mathcal{K}$ is a cactus, we can assume by symmetry that $C'$ lies in $\mathcal{K}[C, x]$.
  If $\varphi(x(f_i)) = x$ for some $i$, then $P_{f_i}$ has no edge in $\mathcal{K}[C, x]$, a contradiction.
  Otherwise, $\varphi(x(f_1))$ and $\varphi(x(f_2))$ lie in distinct components of $\mathcal{K}[C, x]-x$.
  The subpaths from these vertices to $x$ stay in their respective components until reaching $x$.
  However, all vertices of $C'$ other than $x$ lie in a single component, so $C'$ cannot contain an edge of each subpath, a contradiction.
  Thus, $\{f_1,f_2\}$ is a forcing set for $S$.
  % Suppose that two edges $f_1,f_2\in S$ are both $x$-separated and $y$-separated.
  % Then $C$ is the only cycle containing an edge of both $P_{f_1}$ and $P_{f_2}$.
  % \yasu{← この部分がよくわかりませんでした．また，次の結論もLemma 3.2の見た目から直ちに従うように見えないので，なにかもう少し議論があると良さそうです．}
  % By \cref{lem:fgm:path}, any minimum cut containing $f_1$ and $f_2$ is represented by $C$, and hence equals $S$.
  % Thus, $\{f_1,f_2\}$ is a forcing set for $S$.

  % Suppose that no such pair exists.
  % Then no two edges of $S$ are $x$-separated or no two edges of $S$ are $y$-separated.
  % If an $x$-separated pair and a $y$-separated pair both existed in $S$, one could choose an edge from the latter pair that is $y$-separated from both edges of the former pair, which would force the former pair not to be $x$-separated, a contradiction.
  % \yasu{← この部分もよくわかりませんでした．}
  % By symmetry, we may assume that no two edges of $S$ are $x$-separated.
  Suppose that no such pair exists.
  Suppose first that some $f_1, f_2\in S$ are $x$-separated.
  Then they are not $y$-separated, so $\varphi(y(f_1))$ and $\varphi(y(f_2))$ lie in the same component of $\mathcal{K}[C, y]-y$.
  If $\varphi(y(f))$ lies outside this component for some $f\in S$, then $f$ is $y$-separated from both $f_1$ and $f_2$.
  Thus, $f$ is not $x$-separated from either of them, so $\varphi(x(f_1))$, $\varphi(x(f))$, and $\varphi(x(f_2))$ lie in the same component of $\mathcal{K}[C, x]-x$, a contradiction.
  Therefore, all $\varphi(y(f))$ with $f\in S$ lie in the same component, and no two edges of $S$ are $y$-separated.
  Consequently, either no two edges of $S$ are $x$-separated or no two are $y$-separated.
  By symmetry, we may assume the former.
  Since $|S|\ge 2$, none of their endpoints in $A_x$ belongs to $\varphi^{-1}(x)$.
  Thus, some component $D$ of $\mathcal{K}[C,x]-x$ contains $\varphi(x(f))$ for all $f \in S$. 
  Let $B:=\varphi^{-1}(D)$.
  % The two edges incident with $x$ on the cycle leading to $D$ represent the cut induced by $B$, so $\weight(B, V(G)\setminus B)=\lambda$. \yasu{← the cycle leading to $D$ という表現はもう少し正確にしたほうが良さそうです．($D$が全ての $\varphi(x(f))$ を含むことから，目的の閉路が存在するという一文があると良さそうです．)}
  Since $D$ contains $\varphi(x(f))$ for every $f\in S$, some $P_f$ uses an edge joining $x$ and $D$.
  As $\mathcal{K}$ is a cactus, this edge lies on a cycle whose vertices other than $x$ lie in $D$.
  The two edges of this cycle incident with $x$ represent the cut induced by $B$, so $\weight(B,V(G)\setminus B)=\lambda$.
  Since $x(f) \in B$ for every $f \in S$, we have $\weight(B,A_y)=\weight(A_x,A_y)=\lambda$.
  Hence, no edge joins $B$ and $A_x\setminus B$, and no edge joins $A_x\setminus B$ and $A_y$.
  Since $G$ is connected, it follows that $A_x=B$.
  Consequently, $\varphi^{-1}(x)=\emptyset$, and every component of $\mathcal{K}[C,x]-x$ other than $D$ has an empty inverse image.
  No such component exists, since it would define a shore of a minimum cut represented by $\mathcal{K}$.
  Therefore, $x$ is a $2$-junction, contradicting the assumption that $(\mathcal{K},\varphi)$ is reduced.
  Thus, the required pair $f_1,f_2$ exists.
  ~\hfill$\qed$
\end{proofsketch}

% Every minimum cut of $G$ is represented by a pair of edges $e_1,e_2$ contained in a common cycle $C$ of $\mathcal{K}$.
% We distinguish the following three cases:
% (i) $C$ has length at least $3$, and $e_1$ and $e_2$ have no common endpoint;
% (ii) $C$ has length at least $3$, and $e_1$ and $e_2$ have a common endpoint; and
% (iii) $C$ has length $2$.
% By \cref{lem:fgm:nonadj},
% \cref{lem:fgm:adj}, and~\cref{lem:fgm:2cycle}, respectively, the cutset of every minimum cut of $G$ has a forcing set of size at most $2$.

We are now ready to present the algorithm.
\begin{theorem}
  \label{thm:fgm:near-linear}
  \fgmincut{} can be solved in $\bigoh(m\,\mathrm{polylog}\,n)$ time, up to a logarithmic factor in the total edge weight.
\end{theorem}
\begin{proof}
  Let $n:=|V(G)|$ and $m:=|E(G)|$.
  We first construct a cactus representation $(\mathcal{K},\varphi)$ of $\Lambda(G)$.
  The cactus has $\bigoh(n)$ vertices and edges.
  By repeatedly applying \cref{lem:fgm:reduce}, we can obtain a reduced representation in $\bigoh(n)$ time.
  If $k\ge 2$, the answer is yes by \cref{lem:fgm:two}.

  We next consider the case $k=1$.
  \begin{claim}
    \label{clm:fgm:single}
    $G$ has a forcing set of size $1$ if and only if there exists a $2$-cycle of $\mathcal{K}$ with endpoints $x$ and $y$ such that $E_G(\varphi^{-1}(x),\varphi^{-1}(y))\neq\emptyset$.
  \end{claim}
  \begin{claimproof}
    First, suppose that $\mathcal{K}$ has a $2$-cycle $C$ with $V(C) = \{x, y\}$ such that $E_G(\varphi^{-1}(x),\varphi^{-1}(y))\neq\emptyset$.
    Let $e=uv\in E_G(\varphi^{-1}(x),\varphi^{-1}(y))$.
    The path $P_e$ consists of one of the two edges of $C$.
    By \cref{lem:fgm:path}, every minimum cut containing $e$ is represented by $C$.
    Since $C$ has exactly two edges, this minimum cut is unique.
    Hence, $\{e\}$ is a forcing set.

    Conversely, suppose that $\{e\}$ is a forcing set.
    If $P_e$ traverses a cycle $C$ of length at least $3$, then an edge of $C$ used by $P_e$ can be paired with two different edges of $C$ to obtain two distinct minimum cuts containing $e$, a contradiction.
    Thus, every cycle traversed by $P_e$ is a $2$-cycle.

    Moreover, $P_e$ cannot traverse two distinct $2$-cycles, since they represent distinct minimum cuts in a reduced cactus representation, both containing $e$.
    Hence $P_e$ traverses exactly one $2$-cycle $C$, say with endpoints $x$ and $y$.
    Therefore, $e\in E_G(\varphi^{-1}(x),\varphi^{-1}(y))$.
  \end{claimproof}

  By \cref{clm:fgm:single}, given $(\mathcal{K},\varphi)$, we can mark the unordered pair of endpoints of every $2$-cycle of $\mathcal{K}$.
  We then examine every edge $\{u,v\}\in E(G)$ and test whether $\{\varphi(u),\varphi(v)\}$ is one of the marked pairs.
  Thus, the case $k=1$ can be decided in $\bigOh(n+m)$ time.

  Finally, suppose that $k=0$.
  The answer is yes if and only if $G$ has a unique minimum cut.
  In a reduced cactus representation, this holds if and only if $\mathcal{K}$ consists of a single $2$-cycle.
  % Indeed, a cycle of length at least $3$ represents at least two distinct minimum cuts.
  % If every cycle has length $2$ but $\mathcal{K}$ contains at least two cycles, two consecutive cycles represent distinct minimum cuts by the same reducedness argument used above.
  Hence, the case $k=0$ can be decided in $\bigOh(n)$ time as well.
  All operations after constructing the cactus representation take $\bigOh(n+m)$ time, which gives the claimed running time.
\end{proof}
% The proof of \cref{thm:fgm:near-linear} also applies when the target minimum cut is specified.
The proof of \cref{thm:fgm:near-linear} also works for a specified minimum cut.
\begin{corollary}
    % \fgmincut{} with a specified minimum cut can be solved in $\bigoh(m\,\mathrm{polylog}\,n)$ time, up to a logarithmic factor in the total edge weight.
    The same running-time bound holds for \fgmincut{} with a specified minimum cut.
\end{corollary}

\subsection{\NP-completeness of \afgmincut{}}
Next we show that \afgmincut{} is \NP-complete.
Since the membership in \NP\ is clear, we only need to prove \NP-hardness.
To this end, we consider the following closely related problem.

\begin{tcolorbox}
\begin{description}
  \item[Problem] \hgmc{}
  \item[Input: ] A graph $G$ and an integer $k$.
  \item[Goal: ] Determine if $G$ has a set $F\subseteq E(G)$ of size at most $k$ such that every minimum cutset of $G$ intersects $F$.
\end{description}
\end{tcolorbox}

\hgmc{} and \afgmincut{} are similar: While \hgmc{} asks to find a smallest edge subset that hits all cutsets in $\Lambda(G)$, \afgmincut{} asks to find one that hits all but one cutset.

\begin{restatable}[$\star$]{theorem}{hittingGlobalHardness}
  \label{thm:afgm:hgmc:NPC}
  \hgmc{} is \NP-complete even if $G$ is simple.
\end{restatable}
\begin{proofsketch}
  Membership in \NP\ is clear, as every $n$-vertex graph has $\bigoh(n^2)$ minimum cuts, which can be enumerated in polynomial time using its cactus representation.
  We prove \NP-hardness by a reduction from the \NP-complete problem, \uta{}~\cite{FredericksonJ81}: given a tree $T$, a set of paths $\mathcal{P}$ of $T$, and an integer $k$, determine if there are at most $k$ paths in $\mathcal{P}$ that cover all edges of $T$.
  % We prove \NP-hardness by a reduction from the following \NP-complete problem, \uta{}~\cite{FredericksonJ81}.
  % \begin{tcolorbox}
  % \begin{description}
  %   \item[Problem] \uta{}
  %   \item[Input: ] A tree $T$, a set of paths $\mathcal{P}$ of $T$, and an integer $k$.
  %   \item[Goal: ]  Determine if there are at most $k$ paths in $\mathcal{P}$ that cover all edges of $T$.
  % \end{description}
  % \end{tcolorbox}

  Let $(T,\mathcal P,k)$ be an instance of \uta{}.
  We may assume that every edge of $T$ belongs to some path in $\mathcal P$; otherwise, the instance is trivially a no-instance.
  For each $e\in E(T)$, let $c_e$ be the number of paths in $\mathcal P$ containing $e$.
  Set $c_{\max}=\max_{e\in E(T)}c_e$ and $\lambda=2c_{\max}+1$.

  Let $E_{\mathcal P}=\{e_P:P\in\mathcal P\}$, where $e_P$ joins the endpoints of $P$.
  We construct a multigraph $G$ on $V(T)$ with $\lambda-c_e$ parallel copies of each $e\in E(T)$ and all edges in $E_{\mathcal P}$.
  % For $e\in E(T)$, let $A_e,B_e$ be the vertex sets of the components of $T-e$.
  Finally, we apply \cref{lem:transform-unweighted-simple} to $G$.
  % The correctness is proved in the full version~\cite{GimaKMO26full}.
  The correctness is proved in the appendix.
  \hfill$\qed$
\end{proofsketch}

\begin{restatable}[$\star$]{theorem}{antiForcingGlobalHardness}\label{thm:afgm:main}
  \afgmincut{} is \NP-complete even for unweighted simple graphs.
\end{restatable}
\begin{proofsketch}
  Let $(G, k)$ be an instance of \hgmc{} with $\lambda(G) = \lambda$, where $G$ is a simple graph.
  Let $H$ be the graph obtained from $G$ by adding one new vertex $u$ and $\lambda$ parallel edges between $u$ and an arbitrary vertex, say $v$, in $G$.
  Let $S^*$ be the set of parallel edges added as above.

  % The instance $(G,k)$ is a yes-instance of \hgmc{} if and only if $(H,k)$ is a yes-instance of \afgmincut{}.
  % The proof is given in the appendix.
  Finally, we apply \cref{lem:transform-unweighted-simple} to $H$ to obtain an equivalent unweighted simple graph.
  % The correctness is proved in the full version~\cite{GimaKMO26full}.
  The correctness is proved in the appendix.
  ~\hfill$\qed$
\end{proofsketch}

Observe that the construction above can be used as a reduction from \hgmc{} to \afgmincut{} with a specified cut, where $S^*$ is the cutset to be unique.
Hence, we obtain the following corollary.
\begin{corollary}
  \label{thm:afgm:main-scut}
  \afgmincut{} with a specified cut is \NP-complete even for unweighted simple graphs.
\end{corollary}

%%%%%%%%%%%%%%%%%%%%%%%%%%%%%%%%%%%%%%%%%%%%%%%%%%%%%%%%%%%%%%%%%%%%%%%%%%%%%%%%%%%%%%%%%%%%%%%%%%%%%%%%%%%%%%%%%%%%%%%%%%%%%%%%

\section{Min $s$--$t$ Cut}

Our algorithmic and hardness results for \fminstcut{} and \afminstcut{} exploit a representation theorem of minimum $s$--$t$ cuts due to Picard and Queyranne~\cite{PicardQ80}.

\subsection{Representing all minimum $s$--$t$ cuts}
Let $G = (V, E, \weight)$ be a connected edge-weighted graph and let $s, t \in V$ be distinct vertices.
An \emph{$s$--$t$ flow} of $G$ is a function $f\colon V \times V \to \mathbb R$ such that
\begin{itemize}
    \item $f(u, v) = -f(v, u)$ for $u, v \in V$;
    \item $|f(u, v)| \le \weight(uv)$ for $uv \in E$ and $f(u, v) = 0$ for $uv \notin E$;
    \item $\sum_{v \in V}f(u, v) = 0$ for $u \in V \setminus \{s,t\}$.
\end{itemize}
When $s$ and $t$ are clear from the context, we may simply call it a flow.
The \emph{value} of a flow $f$ is defined as $\sum_{v \in V} f(s, v)$.
Let $f$ be a maximum flow in $G$ with the weight function $\weight$.
We construct a directed graph $G_f$ on the same vertex set $V$ such that for $uv \in E$, $G_f$ contains an arc $(u, v)$ if $f(u, v) = \weight(uv)$; an arc $(v, u)$ if $f(v, u) = \weight(uv)$; and two anti-parallel arcs $(u, v)$ and $(v, u)$ if $|f(u, v)| < \weight(uv)$.
Equivalently, $G_f$ is the reverse of the residual digraph of $(G, f)$.
A vertex set $X \subseteq V$ is called \emph{closed} in $G_f$ if there is no arc from $V \setminus X$ to $X$.
The max-flow min-cut theorem and the result of Picard and Queyranne~\cite{PicardQ80} imply the following characterization.
Note that a maximum flow $f$ and the graph $G_f$ can be computed in $\bigoh(nm)$ time~\cite{KingRT94,Orlin13}.

\begin{theorem}[\cite{PicardQ80}]\label{thm:minstcut_representation}
    A cut $\{X, V(G) \setminus X\}$ with $s \in X$ and $t \notin X$ in $G$ is a minimum $s$--$t$ cut of $G$ if and only if $X$ is closed in $G_f$.
\end{theorem}

Another consequence of \cref{thm:minstcut_representation} is that for every directed walk $W$ from $s$ to $t$ in $G_f$, no two (undirected) edges corresponding to arcs in $W$ are simultaneously contained in any minimum $s$--$t$ cut of $G$.

\begin{restatable}[$\star$]{corollary}{minSTCutWalk}\label{cor:minstcut_digraph}
    Let $W$ be an arbitrary directed walk from $s$ to $t$ in $G_f$ and let $W'$ be the undirected walk between $s$ and $t$ in $G$ that corresponds to $W$.
    Then, $W'$ contains exactly one edge of a cutset $S$ for any $S \in \Lambda_{st}(G)$.
\end{restatable}
% \begin{proof}
%     Let $X$ be the $s$-shore of the cutset $S$.
%     By~\cref{thm:minstcut_representation}, $X$ is closed in $G_f$, meaning that the walk $W$ crosses $S$ exactly once.
% \end{proof}

Observe that for each strongly connected component $C$ of $G_f$, either $C \subseteq X$ or $C \cap X = \emptyset$ for every closed set $X \subseteq V(G)$.
An arc in $G_f$ is said to be \emph{relevant} if its endpoints belong to distinct strongly connected components in $G_f$.
The strongly connected component that contains $s$ (resp.~$t$) is denoted by $C_s$ (resp.~$C_t$).
The following observation follows from the assumption that $G$ is connected.
\begin{restatable}[$\star$]{observation}{minSTCutReachability}\label{obs:unique-source-sink}
    For every strongly connected component $C$ of $G_f$ with $C \neq C_s$, there is at least one relevant arc incoming to a vertex in $C$.
    Symmetrically, for every strongly connected component $C$ of $G_f$ with $C \neq C_t$, there is at least one relevant arc outgoing from a vertex in $C$.
    Consequently, every vertex in $G_f$ is reachable from $s$ and can reach $t$.
\end{restatable}
% \begin{proof}
%     If $G_f$ has a strongly connected component $C \neq C_s$ that has no incoming relevant arc, for every edge $uv$ with $u \in C$ and $v \notin C$, we have $f(u, v) = \weight(uv) > 0$.
%     Hence, the total flow value outgoing from $C$ is positive. 
%     On the other hand, as $s \notin C$, the flow conservation condition implies that this flow value is $0$ if $t  \notin C$ and negative if $t \in C$.
%     This is a contradiction.
%     The second statement is symmetric.
% \end{proof}

A vertex $v$ is said to be \emph{free} if $v \notin C_s \cup C_t$.
Note that if an arc $e = (u, v)$ is relevant, there is no opposite arc $(v, u)$ in $G_f$.
Thus, in the following, we do not distinguish the underlying edge of $G$ from the relevant arc in $G_f$.
We say that an arc $e = (u, v)$ \emph{dominates} a vertex $w$ (or $w$ \emph{is dominated by} $e$) if there is a walk from $s$ to $t$ in $G_f$ passing through both $w$ and $e$.

\begin{restatable}[$\star$]{lemma}{forcingMinSTCutCharacterization}\label{lem:minstcut:forcing}
    Let $X$ be the $s$-shore of a minimum $s$--$t$ cut of $G$ and let $F \subseteq \cutset_G(X)$.
    Then, $F$ is a forcing set for $\Lambda_{st}(G)$ if and only if every free vertex is dominated by some edge in $F$ in $G_f$.
\end{restatable}

\subsection{\NP-hardness}
In this section, we prove that \fminstcut{} and \afminstcut{} are both \NP-complete.
\begin{theorem}\label{thm:fmst:NPC} 
    \fminstcut{} is \NP-complete even for unweighted graphs.
\end{theorem}
The problem belongs to \NP{} due to \cref{lem:minstcut:forcing}.
We prove the theorem by showing a reduction from \textsc{Vertex Cover}.

\paragraph{Construction.}
Let $H = (V, E)$ be a graph with $n = |V|$ and $m = |E|$, and let $k$ be a positive integer with $k \le n$.
Since \textsc{Vertex Cover} is \NP-complete even when the input graph is 4-regular~\cite{FrickeHJ98}, we assume that $H$ is 4-regular.
We then construct an edge-weighted graph $G$ as follows.
Let $M = 2n + 1$.
For each vertex $v_i \in V$, $G$ has three vertices $x_i, y_i, z_i$ forming a path with an internal vertex $z_i$.
For each edge $e_j \in E$, $G$ has a set $W_j$ of $M$ vertices.
For each vertex $w \in W_j$, we add an edge between $z_i$ and $w$ if $e_j$ is incident to $v_i$ in $H$.
Finally, we add two vertices $s$ and $t$ in such a way that $s$ is adjacent to all $x_i$'s and all $y_i$'s, and $t$ is adjacent to all vertices in $W_j$ for all $j$.
Let $W = \bigcup_{1 \le j \le m} W_j$ and, for $1 \le i \le n$, let $B_i = \{sx_i, sy_i, x_iz_i, y_iz_i\}$.
See \cref{fig:fmst} for an illustration.
\begin{figure}[t]
    \centering
    \includegraphics[width=0.50\linewidth]{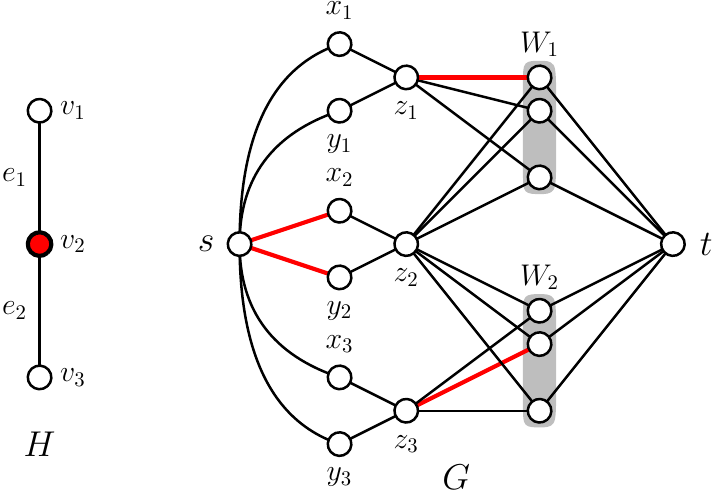}
    \caption{The construction of $G$ in \cref{thm:fmst:NPC}, using a graph $H$ that is not $4$-regular for simplicity.
    % The construction of $G$ from a graph $H$ in the proof of \cref{thm:fmst:NPC}. 
    The red vertex corresponds to a vertex cover $C$ of $H$ and the red edges illustrate how a forcing set is constructed from $C$.
    % the set of red edges corresponds to a forcing set constructed from $C$.
    }
    \label{fig:fmst}
\end{figure}
We set the weight of edge $e$ in $G$ as
\begin{align*}
    \weight(e) = 
    \begin{cases}
        % 2M & e = sx_i \text{ or } e = sy_i\\
        % 2M & e = x_iz_i \text{ or } e = y_iz_i\\
        1 & e = z_iw \text{ for some } w \in W\\
        2 & e = wt \text{ for some } w \in W\\
        2M & \text{ otherwise }
    \end{cases}.
\end{align*}

Note that every $z_i$ is adjacent to exactly $4M$ vertices of $W$ as $H$ is $4$-regular.
Observe that the maximum $s$--$t$ flow $f$ is uniquely determined, with flow value $4nM$, that is, each edge $e$ of $G$ is saturated in the left-to-right direction.
Moreover, $G_f$ is obtained by orienting all edges in $G$ from left to right in \cref{fig:fmst}.
This implies that all vertices except for $s$ and $t$ are free, and all arcs are relevant in $G_f$.
In the following proof, we may interchangeably use $G$ and $G_f$, and when we refer to an edge in $G$, it may also refer to the corresponding arc in $G_f$, and vice-versa.
The correctness of the reduction is justified by the following lemma. 

\begin{lemma}\label{lem:fmst:NPC:correctness}
    $H$ has a vertex cover of size at most $k$ if and only if $G$ has a forcing set for $\Lambda_{st}(G)$ of size at most $n + k$. 
\end{lemma}
\begin{proof}
    Suppose first that $H$ has a vertex cover $C$ of size at most~$k$.
    Let $F$ be the set of edges defined as follows.
    For each $v_i \in C$, $F$ includes two edges $sx_i$ and $sy_i$.
    For each $v_i \notin C$, $F$ includes a single edge $z_iw$ for an arbitrarily chosen $w \in W_j$ for some edge $e_j$ incident to $v_i$ in $H$.
    As $|C| \le k$, we have $|F| = 2|C| + n - |C| \le n + k$, and hence it suffices to show that $F$ is a forcing set.

    Let $X = \{s\} \cup \{x_i, y_i, z_i : v_i \notin C\}$.
    Since there are no arcs from $V(G) \setminus X$ to $X$ in $G_f$, by~\cref{thm:minstcut_representation}, $\{X, V(G) \setminus X\}$ is a minimum $s$--$t$ cut of $G$ and, moreover, we have $F \subseteq \cutset_G(X)$.
    We then claim that $F$ is a forcing set for $\Lambda_{st}(G)$.
    For $v_i \in V$, the vertices $x_i, y_i, z_i$ are dominated by arcs in $F$, as $F$ contains $sx_i$ and $sy_i$ if $v_i \in C$; $F$ contains $z_iw$ for some $w \in W$ if $v_i \notin C$.
    Moreover, for $e_\ell \in E$ with $e_\ell = v_iv_j$, $F$ contains at least one edge in $B_i \cup B_j$ since $C$ is a vertex cover of $H$.
    Thus, the vertices in $W_\ell$ are dominated by arcs in $F$, and hence, by~\cref{lem:minstcut:forcing}, $F$ is a forcing set.

    Suppose next that $G$ has a forcing set $F$ of size at most $n + k$.
    Let $\{X, V(G) \setminus X\}$ be the unique minimum $s$--$t$ cut such that $F \subseteq \cutset_{G}(X)$.
    We first observe that for each edge $e_\ell \in E$ with $e_\ell = v_iv_j$, $F \cap (B_i \cup B_j) \neq \emptyset$.
    This can be seen as follows.
    Since $|W_\ell| \ge 2n + 1 > n + k$, there is a vertex $w \in W_\ell$ such that no edge incident to it is contained in $F$.
    As $w$ is free and dominated by some edge in $F$ in $G_f$ due to \cref{lem:minstcut:forcing}, $F$ contains at least one edge in $B_i \cup B_j$.
    Suppose that $B_i \cap F \neq \emptyset$.
    We can assume that $e \in B_i \cap F$ is incident to $x_i$.
    Let $f \in F$ be an edge dominating $y_i$. 
    If $f \notin B_i$, $f$ belongs to a path $P$ from $z_i$ to $t$ in $G_f$.
    By concatenating the path from $s$ to $z_i$ through $x_i$ with $P$, we have a path from $s$ to $t$ containing at least two edges in $F$, contradicting \cref{cor:minstcut_digraph}.
    
    Now, we define a vertex set $C = \{v_i : |B_i \cap F| \ge 2\}$.
    As observed above, $C$ is a vertex cover of $H$.
    For each $v_i \in V(H) \setminus C$, it holds that $B_i \cap F = \emptyset$.
    As $x_i$ and $y_i$ are free in $G_f$, $F$ contains at least one edge $f_i$ incident to some $w \in W$.
    Moreover, this edge $f_i$ dominates neither $x_j$ nor $y_j$ for $v_j \notin C$ with $j \neq i$, as $v_i$ and $v_j$ are not adjacent to each other in $H$. 
    Thus, for each $v_i \in V(H) \setminus C$, $F$ contains a distinct private edge $f_i$.
    Since $F$ contains at least two edges in $B_i$ for $v_i \in C$ and at least one private edge $f_i$ for $v_i \notin C$, we have $2|C| + n - |C| \le |F| \le n + k$.
    Therefore, $C$ is a vertex cover of $H$ of size at most $k$.
\end{proof}

This proves that \fminstcut{} is \NP-complete.
To make the input graph unweighted, we simply apply the reduction in \cref{lem:transform-unweighted-simple}.
This completes the proof of \cref{thm:fmst:NPC}.

\begin{restatable}[$\star$]{theorem}{antiForcingMinSTCutHardness}\label{thm:afmst:NPC}
    \afminstcut{} is \NP-complete even for unweighted graphs.
\end{restatable}

\subsection{Forcing/Anti-forcing a specified minimum $s$--$t$ cut}

% In the following, 
In stark contrast to the above two hardness results, we show that \fminstcut{} and \afminstcut{} with a specified cut can be respectively reduced to the problems of finding a minimum $T$-edge cover and a minimum-weight arborescence, which can be solved in polynomial time~\cite{lawler1976combinatorial,GabowGST86}.
\begin{restatable}[$\star$]{theorem}{forcingSpecifiedMinSTCut}\label{thm:fmstsc:poly}
    \fminstcut{} with a specified cut is solvable in $\bigoh(nm)$ time, where $n$ and $m$ are the numbers of vertices and edges in $G$.
\end{restatable}
\begin{restatable}[$\star$]{theorem}{antiForcingSpecifiedMinSTCut}\label{thm:afmstsc:poly}
    \afminstcut{} with a specified cut is solvable in $\bigoh(nm)$ time, where $n$ and $m$ are the numbers of vertices and edges in $G$.
\end{restatable}
%\printbibliography

%%%%%%%%%%%%%%%%%%%%%%%%%%%%%%%%%%%%%%%%%%%%%%%%%%%%%%%%%%%%%%%%%%%%%%%%%%%%%%%%%%%%%%%%%%%%%%%%%%%%%%%%%%%%%%%%%%%%%%%%%%%%%%%%

% ---- Bibliography ----
\bibliographystyle{splncs04}
\bibliography{references}

%%%%%%%%%%%%%%%%%%%%%%%%%%%%%%%%%%%%%%%%%%%%%%%%%%%%%%%%%%%%%%%%%%%%%%%%%%%%%%%%%%%%%%%%%%%%%%%%%%%%%%%%%%%%%%%%%%%%%%%%%%%%%%%%

\clearpage
\appendix
\renewcommand{\theHsection}{\Alph{section}}
% Restatements do not advance the theorem counter, so appendix claims
% use an independent counter within each section.
\counterwithout*{myclaim}{theorem}
\counterwithin{myclaim}{section}
\renewcommand{\theHmyclaim}{\theHsection.\arabic{myclaim}}

\section{Preliminaries}

\transformUnweightedSimple*
\begin{proof}
    First, replace each edge $e$ with $\weight(e)$ parallel unweighted edges.
    This preserves the value and the shores of every cut.

    Let $H$ be the resulting unweighted multigraph and let $m=|E(H)|$.
    Replace each vertex $v$ by a clique $K_v$ of size $m+2$, and represent each edge $uv$ of $H$ by an edge between distinct vertices of $K_u$ and $K_v$.
    Any cut splitting a clique has size at least $m+1$, whereas every cut corresponding to a cut of $H$ has size at most $m$.
    Hence, no minimum cut splits a clique, and the claimed correspondence follows.
    
    For minimum $s$--$t$ cuts, apply the same construction and choose one vertex of $K_s$ and one vertex of $K_t$ as the new terminals $s$ and $t$, which proves the latter claim.
\end{proof}
Note that if all edge weights are polynomially bounded, then the transformation can be performed in polynomial time.
These transformations also preserve the minimum cardinalities of forcing and anti-forcing sets.
Indeed, in the first transformation, each original edge in such a set can be replaced by an arbitrary edge in the corresponding parallel bundle, and in the second transformation by its corresponding edge between the cliques.
Conversely, the transformation can be reversed by mapping these edges back to the original edges and discarding clique edges, which belong to no minimum cut.

\section{\fgmincut{} and \afgmincut{}}

\forcingGlobalReduce*
\begin{proof}
  Let $\pi\colon V(\mathcal{K})\to V(\mathcal{K}')$ be the natural projection onto the cactus $\mathcal{K}'$ obtained by the stated identification, and put $\varphi':=\pi\circ\varphi$.
  Let $C'$ be the other cycle containing $z$, and let $e_1,e_2$ be the two edges of $C'$ incident with $z$.
  The cut represented by the two edges of $C$ and the cut represented by $\{e_1,e_2\}$ differ only in the placement of $z$.
  Since $\varphi^{-1}(z)=\emptyset$, they induce the same cut of $G$.
  Therefore, the cut represented by $C$ is redundant.
  For every cut represented by a cycle different from $C$, the vertices $z$ and $u$ lie on the same side, and hence identifying them does not change its inverse image.
  Thus, $(\mathcal{K}',\varphi')$ is a cactus representation of $G$.
\end{proof}

\forcingGlobalTwo*
\begin{proof}
  We prove Case 2, which was omitted from the main text, using the same notation.
  Let $S\in\Lambda(G)$ be represented by $e_1=xz$ and $e_2=yz$ on a cycle $C$ of length at least $3$.
  By \cref{lem:fgm:half}, we have $\weight(A_x,A_z)=\weight(A_y,A_z)=\lambda/2>0$ and $S=E_G(A_z,A_x\cup A_y)$.
  
  Suppose first that some edge joins $A_x\cup A_y$ and $\varphi^{-1}(z)$.
  By symmetry, assume that this edge has its other endpoint in $A_x$.
  Choose $f_x\in E_G(A_x,\varphi^{-1}(z))$ and $f_y\in E_G(A_y,A_z)$.
  % Since $P_{f_x}$ ends at $z$ and the two paths approach $C$ through distinct branches, $C$ is the only cycle containing an edge of both $P_{f_x}$ and $P_{f_y}$. 
  % \yasu{先頭から読むとthe two paths がどれを指すのかわからなかったのと，approach $C$という表現もあまり正確でないと思います．また，branch という語は未定義だと思います．}
  % By \cref{lem:fgm:path}, any minimum cut containing $f_x$ and $f_y$ is represented by the two edges of $C$, and they must be $xz$ and $yz$. 
  By \cref{lem:fgm:path}, $P_{f_x}$ uses $xz$ and has all its other edges in $\mathcal{K}[C,x]$.
  Similarly, $P_{f_y}$ uses $yz$ and has all its other edges in $\mathcal{K}[C,y]\cup\mathcal{K}[C,z]$.
  Every cycle other than $C$ lies in a single component of $\mathcal{K}-E(C)$, so $C$ is the only cycle containing an edge of each of $P_{f_x}$ and $P_{f_y}$.
  By \cref{lem:fgm:path}, any minimum cut containing $f_x$ and $f_y$ must therefore be represented by $xz$ and $yz$.
  Thus, $\{f_x,f_y\}$ is a forcing set for $S$.

  We may therefore assume that there is no edge from $A_x\cup A_y$ to $\varphi^{-1}(z)$.
  Let $C_1,\ldots,C_k$ be the cycles other than $C$ that contain $z$.
  For each $i$, let $D_i$ be the component of $\mathcal{K}-z$ containing $V(C_i)\setminus\{z\}$, and let $B_i:=\varphi^{-1}(D_i)$.
  Thus, $A_z=\varphi^{-1}(z)\cup B_1\cup\cdots\cup B_k$.
  Define $I_x:=\{i\colon E_G(A_x,B_i)\neq\emptyset\}$ and $I_y:=\{i\colon E_G(A_y,B_i)\neq\emptyset\}$.
  Both sets are nonempty, as $\weight(A_x,A_z)=\weight(A_y,A_z)>0$ and no edge joins $A_x\cup A_y$ to $\varphi^{-1}(z)$. 
  % \yasu{←理由が欲しいです．}
  If there are distinct $i\in I_x$ and $j\in I_y$, choose $f_x\in E_G(A_x,B_i)$ and $f_y\in E_G(A_y,B_j)$.
  % Outside $C$, the two paths enter the $z$-side through the distinct components $D_i$ and $D_j$, respectively, while their other sides lie in the distinct components $\mathcal{K}[C, x]$ and $\mathcal{K}[C, y]$. \yasu{the two paths が何を指しているかわからなかったのと，(それが理由か) 全体的に言いたいことがわかりにくいと思います．}
  % Since $i\neq j$, no cycle other than $C$ can contain an edge of both $P_{f_x}$ and $P_{f_y}$. 
  Outside $C$, the edges of $P_{f_x}$ lie in $\mathcal{K}[C,x]\cup\mathcal{K}[C,z]$, and those of $P_{f_y}$ lie in $\mathcal{K}[C,y]\cup\mathcal{K}[C,z]$.
  Thus, any cycle other than $C$ containing an edge of each path must lie in $\mathcal{K}[C,z]$.
  However, the vertices of $P_{f_x}$ and $P_{f_y}$ in $\mathcal{K}[C,z]-z$ lie in the distinct components $D_i$ and $D_j$ of $\mathcal{K}-z$, respectively.
  A cycle cannot contain vertices in both components, so $C$ is the only cycle containing an edge of each path.
  By \cref{lem:fgm:path}, any minimum cut containing $f_x$ and $f_y$ must therefore be represented by $xz$ and $yz$, and hence equals $S$.
  Thus, $\{f_x,f_y\}$ is a forcing set for $S$.
  
  It remains to consider the case where no such distinct indices exist.
  Then $I_x=I_y=\{i^\ast\}$ for some $i^\ast$.
  Consequently, the endpoint in $A_z$ of every edge of $S$ belongs to $B_{i^\ast}$.
  The two edges of $C_{i^\ast}$ incident with $z$ represent the cut induced by $B_{i^\ast}$, and hence $\weight(B_{i^\ast}, V(G)\setminus B_{i^\ast})=\lambda=\weight(B_{i^\ast},V(G)\setminus A_z)$.
  It follows that there is no edge from $B_{i^\ast}$ to $A_z\setminus B_{i^\ast}$.
  There is also no edge from $A_z\setminus B_{i^\ast}$ to $V(G)\setminus A_z$.
  Since $G$ is connected, we obtain $A_z= B_{i^\ast}$.
  In particular, $\varphi^{-1}(z)=\emptyset$ and $B_j=\emptyset$ for every $j\neq i^\ast$.
  No such $C_j$ exists, as otherwise its two edges incident with $z$ would represent a minimum cut having $B_j = \emptyset$ as one of the shores.
  Thus, $z$ belongs to exactly $C$ and $C_{i^\ast}$, that is, it is a $2$-junction in $\mathcal{K}$.
  Since $(\mathcal{K},\varphi)$ is reduced, $C_{i^\ast}$ has length at least $3$.

  Let $p$ and $q$ be the two neighbors of $z$ in $C_{i^\ast}$.
  Note that $p$ and $q$ are distinct.
  Let $R_p:=\varphi^{-1}(\mathcal{K}[C_{i^\ast},p])$ and $R_q:=\varphi^{-1}(\mathcal{K}[C_{i^\ast},q])$.
  Applying \cref{lem:fgm:half} to $C$, we obtain $\weight(A_x,A_z)=\weight(A_y,A_z)=\lambda/2$.
  % By applying \cref{lem:fgm:half} to $C_{i^\ast}$, neither $A_x$ nor $A_y$ has an edge to a branch of $C_{i^\ast}$ corresponding to a vertex nonadjacent to $z$. \yasu{←branchという語が未定義．また，$x$と$y$は$C$上の頂点ですが，Lemma 3.2を$C_{i^\ast}$へ適用して得られる結論を議論しているところがよくわかりませんでした．}
  % Since $\varphi^{-1}(z) = \emptyset$, every edge from $A_x$ or $A_y$ to $A_z$ therefore has its endpoint in $R_p\cup R_q$. \yasu{← every edge from ... は単に，every edge in $S$で良い？}
  As $A_z=B_{i^\ast}$, we have $\varphi^{-1}(\mathcal{K}[C_{i^\ast},z])=V(G)\setminus A_z$.
  Since $p$ and $q$ are the only neighbors of $z$ on $C_{i^\ast}$, \cref{lem:fgm:half} applied to $C_{i^\ast}$ implies that every edge leaving $V(G)\setminus A_z$ has its other endpoint in $R_p\cup R_q$.
  Consequently, every edge in $S$ joins $A_x\cup A_y$ to $R_p\cup R_q$.
  Hence, $\weight(A_x,R_p\cup R_q)=\weight(A_y,R_p\cup R_q)=\lambda/2$.
  % Similarly, by applying \cref{lem:fgm:half} to $C$, every edge from $R_p$ or $R_q$ to $V(G)\setminus A_z$ has its endpoint in $A_x\cup A_y$. \yasu{← from $R_p$ or $R_q$ to $V(G)\setminus A_z$ は between $R_p \cup R_q$ and $V(G)\setminus A_z$ということでしょうか？ (union で書くほうが構造が明確だと思います．)}
  % Thus, $\weight(A_x\cup A_y,R_p)=\weight(A_x\cup A_y,R_q)=\lambda/2$.
  Every edge between $R_p\cup R_q$ and $V(G)\setminus A_z$ belongs to $S=E_G(A_z,A_x\cup A_y)$.
  Applying \cref{lem:fgm:half} to $C_{i^\ast}$ gives $\weight(V(G)\setminus A_z,R_p)=\weight(V(G)\setminus A_z,R_q)=\lambda/2$.
  Therefore, $\weight(A_x\cup A_y,R_p)=\weight(A_x\cup A_y,R_q)=\lambda/2$.
  % Applying the same lemma to $C_{i^\ast}$, we obtain $\weight(A_x\cup A_y,R)=\weight(A_x\cup A_y,S)=\lambda/2$.
  These equalities imply that $\weight(A_x,R_p)=\weight(A_y,R_q)$ and $\weight(A_x,R_q)=\weight(A_y,R_p)$.
  Since $\weight(A_x,A_z)=\lambda/2>0$, we can therefore choose either $f_x\in E_G(A_x,R_p)$ and $f_y\in E_G(A_y,R_q)$, or $f_x\in E_G(A_x,R_q)$ and $f_y\in E_G(A_y,R_p)$.
  By symmetry, assume the former.
  By \cref{lem:fgm:path}, we have $E(P_{f_x})\cap E(C)=\{xz\}$, $E(P_{f_y})\cap E(C)=\{yz\}$, $E(P_{f_x})\cap E(C_{i^\ast})=\{zp\}$, and $E(P_{f_y})\cap E(C_{i^\ast})=\{zq\}$. 
  Outside $C\cup C_{i^\ast}$, the edges of $P_{f_x}$ lie in $\mathcal{K}[C,x]\cup\mathcal{K}[C_{i^\ast},p]$, and those of $P_{f_y}$ lie in $\mathcal{K}[C,y]\cup\mathcal{K}[C_{i^\ast},q]$.
  These four subgraphs are distinct components of $\mathcal{K}-(E(C)\cup E(C_{i^\ast}))$, so no cycle other than $C$ and $C_{i^\ast}$ contains an edge of each path.
  % Moreover, no cycle other than $C$ and $C_{i^\ast}$ contains an edge of both $P_{f_x}$ and $P_{f_y}$.
%   Let $F'$ be a minimum cut containing both $f_x$ and $f_y$, and let $e'_1,e'_2$ be the two edges of a cycle $C'$ representing $F'$.
%   By \cref{lem:fgm:path}, $C'$ is either $C$ or $C_{i^\ast}$.
%   If $C'=C$, then $\{e'_1,e'_2\}=\{\{x,z\},\{y,z\}\}$; if $C'=C_{i^\ast}$, then $\{e'_1,e'_2\}=\{\{z,p\},\{z,q\}\}$.
%   Both pairs represent $F$ since $A_z=B_{i^\ast}$.
%   Hence, $F'=F$, and $\{f_x,f_y\}$ is a forcing set for $F$.
  Suppose that there exists $S'\in\Lambda(G)\setminus\{S\}$ such that $\{f_x,f_y\}\subseteq S'$.
  Let $e'_1,e'_2$ be the two edges of a cycle $C'$ representing $S'$.
  By \cref{lem:fgm:path}, $C'$ contains an edge from each of $P_{f_x}$ and $P_{f_y}$.
  Hence, either $C' = C$ or $C' = C_{i^\ast}$.
  If $C'=C$, then $\{e'_1,e'_2\}=\{xz,yz\}$.
  If $C'=C_{i^\ast}$, then $\{e'_1,e'_2\}=\{zp,zq\}$.
  Since $A_z=B_{i^\ast}$, both pairs represent $S$, a contradiction.
  Therefore, $\{f_x,f_y\}$ is a forcing set for $S$.
\end{proof}

\hittingGlobalHardness*
The membership in \NP\ is clear.
We prove \NP-hardness by a reduction from the following \NP-complete problem, \uta{}~\cite{FredericksonJ81}.
\begin{tcolorbox}
\begin{description}
\item[Problem] \uta{}
\item[Input: ] A tree $T$, a set of paths $\mathcal{P}$ of $T$, and an integer $k$.
\item[Goal: ] Determine if there are at most $k$ paths in $\mathcal{P}$ that cover all edges of $T$.
\end{description}
\end{tcolorbox}
%\uta{} is \NP-complete even when the diameter of $T$ is $4$~\cite{FredericksonJ81}.

Let $(T, \mathcal{P}, k)$ be an instance of \uta{}.
If some edge of $T$ is contained in no path of $\mathcal{P}$, the instance is trivially a no-instance.
Hence, we may assume that each edge of $T$ is included in at least one path in $\mathcal{P}$.
For each $e\in E(T)$, let $c_e$ be the number of paths in $\mathcal{P}$ that contain $e$.
Let $c_{\max} = \max_{e \in E(T)} c_e$ and $\lambda = 2c_{\max}+1$.

Let $E_{\mathcal{P}} = \{e_P = uv : P\in \mathcal{P},\, u \,\text{and}\, v \, \text{are the endpoints of}\, P\}$.
We construct a multigraph $G$ on vertex set $V(T)$ by adding $\lambda-c_e$ parallel copies of each edge $e \in E(T)$, together with all edges in $E_{\mathcal{P}}$.
For each edge $e\in E(T)$, let $A_e, B_e$ be the vertex sets of the two connected components of $T-e$.

\begin{lemma}
\label{lem:afgm:hgmc:min-cut}
It holds that $\lambda(G) = \lambda$ and $\Lambda(G) = \{E_G(A_{e}, B_{e}) : e\in E(T)\}$.
\end{lemma}
\begin{proof}
Observe that $E_G(A_e, B_e)$ consists of $\lambda-c_e$ parallel edges corresponding to $e$ and $c_e$ edges in $E_{\mathcal{P}}$ corresponding to the paths in $\mathcal{P}$ that contain $e$.
It follows that $|E_G(A_{e}, B_{e})|=\lambda$.
    
Now, let $\{X,Y\}$ be a cut of $G$ such that $\{X,Y\}\neq\{A_e,B_e\}$ for every $e\in E(T)$.
Then, there exist two edges $e_1, e_2 \in E(T)$ that connect $X$ and $Y$.
Hence, $E_G(X,Y)$ contains the $\lambda-c_{e_i}$ parallel edges corresponding to $e_i$ for each $i\in\{1,2\}$.
Therefore, $|E_{G}(X,Y)| \ge (\lambda - c_{e_{1}}) + (\lambda - c_{e_{2}}) \ge 2\lambda - 2 c_{\max} = \lambda + 1$.
\end{proof}

\begin{lemma}
\label{lem:afgm:hgmc:iff}
$(T,\mathcal{P},k)$ is a yes-instance of \uta{} if and only if $(G,k)$ is a yes-instance of \hgmc{}.
\end{lemma}
\begin{proof}
Let $\mathcal{P}'\subseteq \mathcal{P}$ be a set of at most $k$ paths that cover all edges of $T$.
Let $E_{\mathcal{P}'} = \{e_P\in E_{\mathcal{P}} : P\in \mathcal{P}'\}$.
We show that $E_{\mathcal{P}'}$ is a solution for the instance $(G, k)$.
By \cref{lem:afgm:hgmc:min-cut}, it suffices to show that $E_{\mathcal{P}'}\cap E_G(A_e,B_e)\neq\emptyset$ for every $e\in E(T)$.

Fix $e\in E(T)$.
Since $\mathcal{P}'$ covers all edges of $T$, there is a path $P \in \mathcal{P}'$ that contains $e$, and thus $e_P\in E_{\mathcal{P}'}$.
Since $P$ contains $e$, one of the endpoints of $e_P$ belongs to $A_e$ and the other belongs to $B_e$.
This implies $e_P\in E_G(A_e, B_e)$.

Conversely, let $F$ be a solution for $(G, k)$.
Suppose that $F$ contains a parallel copy $\tilde e$ corresponding to some $e\in E(T)$. 
Since every edge of $T$ belongs to some path in $\mathcal{P}$, there is a path $P\in\mathcal{P}$ containing $e$. 
The edge $\tilde e$ belongs to $E_G(A_e,B_e)$, but to no $E_G(A_f,B_f)$ with $f\neq e$.
Since $P$ contains $e$, the endpoints of $P$ belong to different components of $T-e$, and hence $e_P\in E_G(A_e, B_e)$.
Thus, every minimum cut hit by $\tilde e$ is also hit by $e_P$.
Therefore, replacing $\tilde e$ with $e_P$ preserves feasibility.
By repeatedly applying this replacement, we can assume that $F\subseteq E_{\mathcal{P}}$.
By \cref{lem:afgm:hgmc:min-cut}, $F\cap E_G(A_e,B_e)\neq\emptyset$ for every $e\in E(T)$.

Let $\mathcal{P}'=\{P\in\mathcal{P} : e_P\in F\}$.
Clearly, $|\mathcal{P}'|\le k$.
It remains to show that $\mathcal{P}'$ covers every edge of $T$.
Fix $e\in E(T)$, and let $e_P\in F\cap E_G(A_e,B_e)$.
Since the endpoints of $e_P$ belong to different sides of $\{A_e,B_e\}$, $P$ contains $e$.
\end{proof}

By \cref{lem:transform-unweighted-simple}, $G$ can be transformed in polynomial time into a simple graph $G'$ whose minimum cuts are in one-to-one correspondence with those of $G$.
This transformation also preserves feasible solutions of \hgmc{} without increasing their size.
Therefore, \hgmc{} remains \NP-hard even when the input graph is simple.

\antiForcingGlobalHardness*
\begin{proof}
  Let $(G, k)$ be an instance of \hgmc{} with $\lambda(G) = \lambda$, where $G$ is a simple graph.
  Let $H$ be the graph obtained from $G$ by adding one new vertex $u$ and $\lambda$ parallel edges between $u$ and an arbitrary vertex, say $v$, in $G$.
  Let $S^*$ be the set of parallel edges added as above.
  
  \begin{claim}
    \label{clm:afgm:min-cut-H}
    It holds that $\Lambda(H)=\Lambda(G)\cup\{S^*\}$.
  \end{claim}
  \begin{claimproof}
    Since $|S^*|=\lambda$ and $\cutset_{H}(\{u\}) = S^*$, we have $\lambda(H)\le\lambda$.
    Consider an arbitrary cutset $S$ of $H$.
    If $S$ separates $u$ and $v$, then it contains all $\lambda$ parallel edges between them, and hence $|S|\ge\lambda$.
    Otherwise, $u$ and $v$ lie on the same shore of $S$, and $S$ corresponds to a cutset of $G$, implying that $|S|\ge\lambda$.
    Thus, $\lambda(H)=\lambda$.
    Moreover, equality holds precisely when $S = S^*$ or $S \in \Lambda(G)$.
  \end{claimproof}
    
  Now we show that $(G,k)$ is a yes-instance of \hgmc{} if and only if $(H,k)$ is a yes-instance of \afgmincut{}.

  First, let $F$ be a solution for $(G,k)$.
  Then $F$ hits all the cutsets in $\Lambda(G)$.
  Since $F\subseteq E(G)$, we have $F\cap S^* = \emptyset$.
  By \cref{clm:afgm:min-cut-H}, every $S \in \Lambda(H)$ other than $S^*$ is the cutset of a minimum cut of $G$.
  Thus, $F$ is an anti-forcing set for $\Lambda(H)$.

  Conversely, let $F$ be an anti-forcing set for $\Lambda(H)$ with $F \subseteq E(H) \setminus S$ for some $S \in \Lambda(H)$.
  If $S = S^*$, then $F$ hits all cutsets in $\Lambda(G)$, and we are done.
  Suppose that $S \neq S^*$.
  Then $S \in \Lambda(G)$ and $F \cap S^* \neq \emptyset$.
  Choose an arbitrary edge $e\in S$ and let $F' = (F\setminus S^*)\cup \{e\}$.
  Since $F\cap S^*\neq \emptyset$, we have $|F'|\le |F|$.
  The cutset of every minimum cut of $G$ other than $S$ is hit by $F\setminus S^*$, while $S$ is hit by $e$.
  Therefore, $F'$ hits the cutset of every minimum cut of $G$.

  Finally, apply \cref{lem:transform-unweighted-simple} to $H$.
  Since $H$ is an unweighted multigraph of polynomial size, this yields an equivalent unweighted simple graph while preserving anti-forcing sets and the specified target cut.
\end{proof}

%%%%%%%%%%%%%%%%%%%%%%%%%%%%%%%%%%%%%%%%%%%%%%%%%%%%%%%%%%%%%%%%%%%%%%%%%%%%%%%%%%%%%%%%%%%%%%%%%%%%%%%%%%%%%%%%%%%%%%%%%%%%%%%%

\section{\fminstcut{} and \afminstcut{}}

\subsection{Representing all minimum $s$--$t$ cuts}
\minSTCutWalk*
\begin{proof}
    Let $X$ be the $s$-shore of the cutset $S$.
    By~\cref{thm:minstcut_representation}, $X$ is closed in $G_f$, meaning that the walk $W$ crosses $S$ exactly once.
\end{proof}

\minSTCutReachability*
\begin{proof}
    If $G_f$ has a strongly connected component $C \neq C_s$ that has no incoming relevant arc, for every edge $uv$ with $u \in C$ and $v \notin C$, we have $f(u, v) = \weight(uv) > 0$.
    Hence, the total flow value outgoing from $C$ is positive.
    On the other hand, as $s \notin C$, the flow conservation condition implies that this flow value is $0$ if $t  \notin C$ and negative if $t \in C$.
    This is a contradiction.
    The second statement is symmetric.
\end{proof}

\forcingMinSTCutCharacterization*
\begin{proof}
    Suppose that $F$ is a forcing set for $\Lambda_{st}(G)$ with $F \subseteq \cutset_G(X)$.
    Suppose to the contrary that $G_f$ has a free vertex $v$ that is not dominated by any edge in $F$.
    We can assume that $v \in X$ due to symmetry.
    Let $Z$ be the set of vertices in $X$ that are reachable from $v$ in $G_f[X]$.

    \begin{claim}
        $X \setminus Z$ is closed in $G_f$.
    \end{claim}
    \begin{claimproof}
        Suppose that there is an arc $e$ from $(V(G) \setminus X) \cup Z$ to $X \setminus Z$ in $G_f$.
        Since $X$ is closed, this arc $e$ is directed from $Z$ to $X \setminus Z$.
        Since $Z$ is the set of vertices in $G_f[X]$ that are reachable from $v$, the head of $e$, which is contained in $X \setminus Z$, is reachable, leading to a contradiction.
    \end{claimproof}

    Since $v \notin C_s$, we have $s \notin Z$.
    This implies that, by~\cref{thm:minstcut_representation}, $X \setminus Z$ is an $s$-shore of a minimum $s$--$t$ cut of $G$.
    Since $v$ is not dominated by any edge in $F$, $\cutset_{G}(X\setminus Z)$ contains $F$.
    This contradicts the fact that $F$ is a forcing set.

    Conversely, suppose that $G$ has an edge set $F$ that is contained in $\cutset_G(X)$ and every free vertex is dominated by some arc in $F$.
    Suppose to the contrary that there is a closed $s$-shore $X' \neq X$ that satisfies $F \subseteq \cutset_G(X')$.
    For each $uv \in F \cap \cutset_{G}(X)$, $G_f$ has an arc $(u, v)$, and hence $u \in X$ and $v \notin X$.
    As $uv \in \cutset_G(X')$, we have $u \in X \cap X'$ and $v \notin X \cup X'$.
    Choose an arbitrary vertex $w \in X \symdif X'$.
    Since $C_s$ is contained in both $X$ and $X'$ and $C_t$ is disjoint from both, $w$ is free.
    Then, as $w$ is dominated by some edge $ab \in F$, $G_f$ has a directed walk that contains $w$ and $(a, b)$ with $a \in X$ and $b \notin X$.
    If $w$ occurs before $(a, b)$ in this walk, then the subwalk from $w$ to $a$ enters one of $X$ and $X'$ from outside.
    If $w$ occurs after $(a, b)$, then the subwalk from $b$ to $w$ enters the other set from outside.
    In both cases, at least one of $X$ and $X'$ is not closed in $G_f$, deriving a contradiction.
\end{proof}

% We then turn our attention to anti-forcing sets.
We next establish some properties of anti-forcing sets for minimum $s$--$t$ cuts that will be used in the following proofs.
By~\cref{thm:minstcut_representation}, each closed set $X$ with $s \in X$ and $t \notin X$ corresponds to a minimum $s$--$t$ cut $\{X, V(G) \setminus X\}$ of $G$, and vice-versa.
Moreover, for every anti-forcing set $F \subseteq E(G) \setminus \cutset_G(X)$, it hits $\cutset_G(X')$ for every minimum $s$--$t$ cut $\{X', V(G) \setminus X'\}$ with $X' \neq X$.
Thus, we have the following proposition.

\begin{proposition}\label{apd:prop:minstcut:anti-forcing}
    Let $X$ be the $s$-shore of a minimum $s$--$t$ cut of $G$ and let $F \subseteq E(G) \setminus \cutset_G(X)$.
    Then, $F$ is an anti-forcing set for $\Lambda_{st}(G)$ if and only if $\cutset_G(X') \cap F \neq \emptyset$ for every closed set $X' \subseteq V(G)$ such that $s \in X'$, $t \notin X'$, and $X' \neq X$.
\end{proposition}

This proposition requires $F$ to hit the cutset $\cutset_G(X')$ of any closed set $X'$ distinct from $X$ with $s \in X'$ and $t \notin X'$. 
However, by exploiting the properties of minimum $s$--$t$ cuts, it suffices to hit those with a shore strictly contained in either $X$ or $V(G)\setminus X$.
Before proving this, we observe that, for vertex sets $X, X' \subseteq V(G)$, every arc incoming to $X \cap X'$ is also incoming to $X$ or $X'$.
Thus, we have the following corollary.

\begin{corollary}\label{cor:union-intersection}
    Let $X$ and $X'$ be $s$-shores of minimum $s$--$t$ cuts of $G$.
    Then $X \cap X'$ is also an $s$-shore of a minimum $s$--$t$ cut of $G$. 
\end{corollary}

\begin{lemma}\label{apd:lem:minstcut:anti-forcing}
    Let $X$ be the $s$-shore of a minimum $s$--$t$ cut of $G$ and let $F \subseteq E(G) \setminus \cutset_G(X)$.
    Then, $F$ is an anti-forcing set for $\Lambda_{st}(G)$ if and only if $\cutset_G(X') \cap F \neq \emptyset$ for every closed set $X' \subseteq V(G)$ with $s \in X'$ and $t \notin X'$ satisfying either $X' \subsetneq X$ or $X \subsetneq X'$.
\end{lemma}
\begin{proof}
    The forward implication is clear due to \cref{apd:prop:minstcut:anti-forcing}.
    Thus, suppose that $F$ satisfies $\cutset_{G}(X') \cap F \neq \emptyset$ for every closed set $X'$ satisfying either $X' \subsetneq X$ or $X \subsetneq X'$.
    Let $Y$ be an arbitrary closed set with $s \in Y$ and $t \notin Y$ such that both $X \setminus Y$ and $Y \setminus X$ are nonempty.
    By~\cref{thm:minstcut_representation,cor:union-intersection}, $X \cap Y$ is also closed with $s \in X \cap Y$ and $t \notin X \cap Y$.
    Since $X \cap Y \subsetneq X$, $F$ contains at least one edge from $\cutset_G(X \cap Y)$.
    As $\cutset_G(X) \cap F = \emptyset$, this edge belongs to $E_G(X \cap Y, X \setminus Y)$, and hence it belongs to $\cutset_{G}(Y)$.
    Thus, by~\cref{apd:prop:minstcut:anti-forcing}, $F$ is an anti-forcing set.
\end{proof}

\subsection{NP-hardness}
\antiForcingMinSTCutHardness*
We first observe that the problem belongs to \NP.
To see this, let $G'$ be the graph obtained from $G$ by setting the weight of each edge in $F$ to a sufficiently large number.
Then, it is easy to verify that $F$ is an anti-forcing set for $\Lambda_{st}(G)$ if and only if $\lambda_{st}(G) = \lambda_{st}(G')$ and $G'$ has a unique minimum $s$--$t$ cut.
This condition can be checked in polynomial time.

Next, we show the \NP-hardness.
We prove the theorem by a polynomial-time reduction from \textsc{Vertex Cover}.
Our reduction first produces an edge-weighted graph, which we later transform into an unweighted graph.

\paragraph{Construction.}
Let $(H, k)$ be an instance of \textsc{Vertex Cover}.
We may assume that $0\le k \le |V(H)|$ and that $H$ has no isolated vertices.
We construct a graph $G$ as follows.
Let $n = |V(H)|$, $m = |E(H)|$, and $V(H) = \{v_1, \ldots, v_n\}$.
Let $M = 2n+1$.
For each $i\in [n]$, let $d_i$ be the degree of $v_i$ in $H$.

We first create two vertices $s$ and $t$ in $G$.
For each vertex $v_i\in V(H)$, we construct a vertex gadget as follows.
We create two vertices $r_i$ and $q_i$, and add four edges $sr_i$, $r_iq_i$, $q_it$, and $r_it$.
The edges $sr_i$, $r_iq_i$, and $q_it$ have weight $1$, and the edge $r_it$ has weight $d_iM$.
For each edge $v_iv_j\in E(H)$, we create an edge gadget as follows.
We create $M$ vertices $p_{ij}^1, p_{ij}^2, \ldots, p_{ij}^M$ in $G$.
For each $p_{ij}^h$ ($h\in [M]$), we add four edges $sp_{ij}^h$, $p_{ij}^ht$, $r_ip_{ij}^h$ and $r_jp_{ij}^h$.
The edge $sp_{ij}^h$ has weight $3$, and the other three edges have weight $1$.
See \cref{fig:afminst-cut-reduction} for an illustration of the construction of each gadget.
  
\begin{figure}[t]
  \centering
  \includegraphics[width=0.9\linewidth]{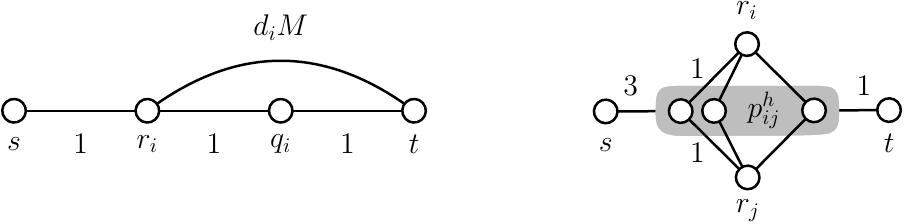}
  \caption{The vertex gadget (left) and the edge gadget (right) in the proof of \cref{thm:afmst:NPC}.
  For visibility, we simply draw edges $sp_{ij}^h$ and edges $p_{ij}^ht$ for $h \in [M]$ as single arcs.}
  \label{fig:afminst-cut-reduction}
\end{figure}

Observe that a maximum flow $f$ with value $n+3Mm$ in $G$ is uniquely determined: $f$ sends one unit along the path $(s, r_i, q_i, t)$ in the vertex gadget for $i \in [n]$ and three units via $p_{ij}^h$ in the edge gadget for $v_iv_j \in E(H)$ and $h \in [M]$, together with shortcut edges $r_it$ and $r_jt$ in the vertex gadgets.
Thus, $G_f$ is obtained from $G$ by orienting each edge of the vertex gadget from left to right in \cref{fig:afminst-cut-reduction} and from $s$ to $p_{ij}^h$ and from $p_{ij}^h$ to the other three vertices in each edge gadget.
A concrete example of $G_f$ is depicted in \cref{fig:agmst_Gf}.
\begin{figure}[t]
  \centering
  \includegraphics[width=0.6\linewidth]{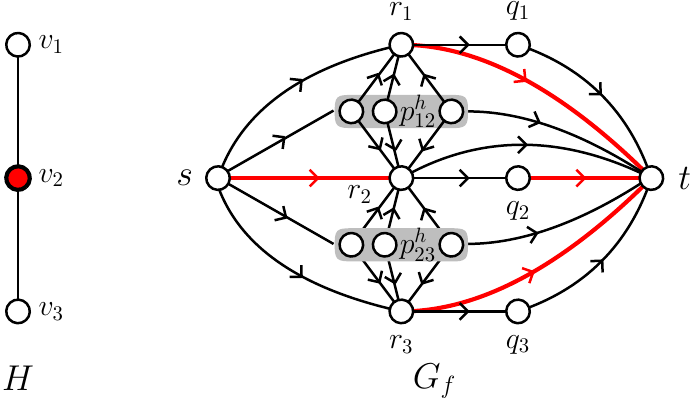}
  \caption{A graph $H$ and the directed graph $G_f$ in the proof of \cref{thm:afmst:NPC}. The red vertex corresponds to a vertex cover $C$ of $H$ and the set of red edges corresponds to an anti-forcing set of $G$.}
  \label{fig:agmst_Gf}
\end{figure}
We now show the correctness of our reduction.

\begin{lemma}
    $H$ has a vertex cover of size at most $k$ if and only if $G$ has an anti-forcing set for $\Lambda_{st}(G)$ of size at most $n+k$.
\end{lemma}
\begin{proof}
  Suppose that $H$ has a vertex cover $C \subseteq V(H)$ with $|C|\le k$.
  Define an $s$--$t$ cut $\{X,V(G) \setminus X\}$ as
  $X \coloneqq \{s\} \cup \{r_i : v_i \in C\} \cup \{p_{ij}^h : v_iv_j\in E(H),\, h \in [M]\}$.
  Since there are no arcs from $V(G) \setminus X$ to $X$ in $G_f$, by~\cref{thm:minstcut_representation}, $\{X,V(G) \setminus X\}$ is indeed a minimum $s$--$t$ cut of $G$.

  Define $F \coloneqq \{r_it : v_i\notin C\}\cup\{sr_i,q_it : v_i\in C\}$.
  Then, we have $|F|=n+|C|\le n+k$ and $F \cap \cutset_{G}(X)=\emptyset$.
  It remains to prove that $\{X, V(G) \setminus X\}$ is the unique minimum $s$--$t$ cut whose cutset is disjoint from $F$, that is, $F$ hits all the cutsets $\cutset_{G}(X') \in \Lambda_{st}(G)$ whose $s$-shore $X'$ is distinct from $X$.
  We can assume that either $X' \subsetneq X$ or $X \subsetneq X'$ due to \cref{apd:lem:minstcut:anti-forcing}.
  
  Suppose that $X' \subsetneq X$.
  If $X \setminus X'$ contains $r_i$ for some $v_i \in C$, then $sr_i \in F$ hits $\cutset_{G}(X')$.
  Suppose otherwise that $X'$ contains $r_i$ for all $v_i \in C$.
  This implies that $X \setminus X'$ contains some $p_{ij}^h$.
  As $C$ is a vertex cover of $H$, at least one of $v_i$ and $v_j$ is contained in $C$.
  Since there are no arcs from $V(G) \setminus X'$ to $X'$ in $G_f$, both $r_i$ and $r_j$ are contained in $V(G) \setminus X'$, which leads to a contradiction.
  
  Suppose next that $X \subsetneq X'$.
  If $X'\setminus X$ contains $r_i$ for some $i$, then $v_i\notin C$, and hence $r_it\in F$ hits $\cutset_G(X')$.
  Otherwise, $X'\setminus X$ contains $q_i$ for some $i$.
  The arc $(r_i, q_i)$ and the closedness of $X'$ imply that $r_i\in X'$.
  Since no vertex $r_j$ belongs to $X'\setminus X$, we have $r_i\in X$, and therefore $v_i\in C$.
  Consequently, $q_it\in F$ hits $\cutset_G(X')$.
  % Then, $X' \setminus X$ contains $r_i$ or $q_i$ for some $i$.
  % If $r_i\in X'\setminus X$, then $v_i\notin C$, and hence $r_it\in S$ hits $\cutset_G(X')$.
  % Suppose otherwise that $q_i\in X'\setminus X$ and $r_i \notin X'$. 
  % Then, the arc $(r_i, q_i)$ is incoming to $X'$.
  % This contradicts the fact that $X'$ is closed in $G_f$.
  
  % If $v_i\notin C$, then $r_it \in S$, and hence $r_i\notin X'$; Claim~\ref{clm:afmst:min-cut-closure} then gives $q_i\in B'$.
  % If $v_i\in C$, then $sr_i,q_it\in F_C$, and hence $r_i\in A'$ and $q_i\in B'$.
  % Thus, $(A',B')$ agrees with $(A_C,B_C)$ on all vertex gadgets.
  % Since $C$ is a vertex cover, every edge $v_iv_j\in E(G)$ has an endpoint in $C$.
  % Consequently, at least one of $r_i$ and $r_j$ lies in $A'$, so Claim~\ref{clm:afmst:min-cut-closure} implies $p_{ij}^h\in A'$ for every $h$.
  % Hence $(A',B')=(A_C,B_C)$, and $F_C$ is an anti-forcing set for $(A_C, B_C)$.\\

  Conversely, suppose that $G$ has an anti-forcing set $F$ with $|F|\le n + k$.
  Let $X$ be the $s$-shore of the unique minimum $s$--$t$ cut of $G$ satisfying $\cutset_{G}(X)\cap F = \emptyset$.
  % Since the cutset of every other minimum $s$--$t$ cut must intersect $F$, every such cut contains an edge of $F$ that is not contained in $\cutset_{G'}(A)$.
  Define $I \coloneqq \{v_i\in V(H) : r_i \notin X\}$.
  By~\cref{thm:minstcut_representation}, $X$ is closed in $G_f$, that is, there are no arcs from $V(G) \setminus X$ to $X$ in $G_f$.

  We first show that $I$ is an independent set of $H$.
  Suppose that $I$ has two adjacent vertices $v_i$ and $v_j$.
  By the definition of $I$, we have $r_i, r_j \notin X$.
  If $X$ contains $p_{ij}^h$ for some $h \in [M]$, then $X \setminus \{p_{ij}^h\}$ is also an $s$-shore of another minimum $s$--$t$ cut of $G$, as the heads of outgoing arcs incident to $p_{ij}^h$, namely $r_i$, $r_j$, and $t$, are not contained in $X$.
  Since $\cutset_{G}(X \setminus \{p_{ij}^h\}) \setminus \cutset_{G}(X) = \{sp_{ij}^h\}$, $F$ must contain $sp_{ij}^h$.
  Similarly, if $X$ does not contain $p_{ij}^h$ for some $h \in [M]$, then $X \cup \{p_{ij}^h\}$ is also an $s$-shore of another minimum $s$--$t$ cut of $G$.
  Since $\cutset_{G}(X \cup \{p_{ij}^h\}) \setminus \cutset_{G}(X) = \{p_{ij}^hr_i, p_{ij}^hr_j, p_{ij}^ht\}$, $F$ must contain at least one of these edges.
  This implies that $F$ contains at least $M$ edges, contradicting the assumption $n + k < M$. 
  Thus, $I$ is an independent set of $H$, and hence $C \coloneqq V(H) \setminus I$ is a vertex cover of $H$.
  
  % moving only $p_{ij}^h$ to the other side preserves both conditions in Claim~\ref{clm:afmst:min-cut-closure}.
  % It therefore produces another minimum $s$--$t$ cut, denoted by $(A_h,B_h)$.
  % Depending on the original side of $p_{ij}^h$, the new cut edges are either $\{sp_{ij}^h\}$ or $\{p_{ij}^ht,r_ip_{ij}^h,r_jp_{ij}^h\}$.
  % The sets $\cutset_{G'}(A_h)\setminus\cutset_{G'}(A)$ for $h\in\{1,\ldots,M\}$ are pairwise edge-disjoint.
  % Since each $(A_h,B_h)$ is a minimum cut different from $(A,B)$, $F$ contains an edge from each of these $M$ sets.
  % Hence $|F|\ge M$.
  % This is impossible because $|F|\le k=n+\ell\le 2n<M=2n+1$.
  % Thus, $X$ is independent, and $C:=V(G)\setminus X$ is a vertex cover.

  Next, we show that $|C| \le k$.
  We first consider a vertex $v_i\in I$.
  Since $r_i \notin X$ and $G_f$ has an arc $(r_i, q_i)$, we have $q_i \notin X$.
  Since $I$ is an independent set of $H$, we have $v_j \notin I$ for every adjacent vertex $v_j$, and hence $r_j \in X$.
  As $G_f$ has an arc $(p_{ij}^h, r_j)$, $p_{ij}^h\in X$ for every $h\in[M]$.
  Consequently, $X \cup \{r_i\}$ is closed in $G_f$, and hence $F$ contains at least one edge in $\cutset_G(X \cup \{r_i\})$ due to \cref{apd:prop:minstcut:anti-forcing}.
  In particular, $F$ contains either $r_iq_i$ or $r_it$.
  
  We next consider a vertex $v_i \notin I$.
  Since $r_i \in X$ and $G_f$ has an arc $(p_{ij}^h, r_i)$, $p_{ij}^h\in X$ for every adjacent vertex $v_j$ and $h\in[M]$.
  If $X$ contains $q_i$, $X \setminus \{q_i\}$ is also an $s$-shore of another minimum $s$--$t$ cut of $G$, as the heads of outgoing arcs incident to $q_i$ are not contained in $X$. 
  As $\cutset_G(X \setminus \{q_i\}) \setminus \cutset_G(X) = \{r_iq_i\}$, $F$ must contain $r_iq_i$.
  Moreover, as $X \setminus \{r_i, q_i\}$ is also closed in $G_f$, $F$ must contain at least one edge of $\{sr_i\}\cup \{p_{ij}^h r_i \colon v_iv_j\in E(H),\ h\in[M]\}$.
  Similarly, if $X$ does not contain $q_i$, both $X \setminus \{r_i\}$ and $X \cup \{q_i\}$ are closed in $G_f$.
  This implies that $F$ must contain at least one edge of $\{sr_i\}\cup \{p_{ij}^h r_i \colon v_iv_j\in E(H),\ h\in[M]\}$ and must also contain $q_it$.

  Therefore, $F$ contains at least one edge for each $v_i \in I$ and at least two edges for each $v_i \in C$.
  These edges are disjoint for distinct $i, j \in [n]$.
  Hence, we have $2|C| + |I| \le |F| \le n + k$, implying that $|C| \le k$.
  % If $q_i\in V(G) \setminus X$, moving only $q_i$ to $X$ produces another minimum cut whose only new cut edge is $q_it$;
  % if $q_i\in X$, moving only $q_i$ to $V(G) \setminus X$ produces another minimum cut whose only new cut edge is $r_iq_i$.
  % Thus, $q_it\in S$ in the first case and $r_iq_i\in S$ in the second case.
  % To obtain one more edge of $S$, in the first case, move only $r_i$ from $X$ to $V(G) \setminus X$, and in the second case, move both $r_i$ and $q_i$ from $X$ to $V(G) \setminus X$.
  % In either case, its set of new cut edges is exactly $S_i$.
  % Therefore, $S\cap S_i\neq\emptyset$.
  % The sets used above are pairwise edge-disjoint over all vertices: for $v_i\in I$ we use $\{r_iq_i,r_it\}$, while for $v_i\notin I$ we use one of the singleton sets $\{q_it\}$ and $\{r_iq_i\}$ together with $S_i$.
  % In particular, $P_i\cap P_j=\emptyset$ for $i\neq j$, since $r_ip_{ij}^h$ and $r_jp_{ij}^h$ are distinct edges even when they belong to the same copy of an edge gadget.
  % It follows that $|S|\ge |I|+2(n-|I|)=n+|C|$.
  % Since $|S|\le n + k$, we obtain $|C|\le k$.
\end{proof}

By \cref{lem:transform-unweighted-simple}, the construction can be made unweighted.
As all edge weights are polynomially bounded in the size of $H$, the transformation can be done in polynomial time.
This completes the proof of \cref{thm:afmst:NPC}.

\subsection{Algorithms}
\forcingSpecifiedMinSTCut*
\begin{proof}
    Let $G = (V, E, \weight)$ be an input graph and let $\{X, V\setminus X\}$ be a given minimum $s$--$t$ cut of $G$ with $s \in X$.
    We can assume that $G$ is connected.
    Let $f$ be an arbitrary maximum flow in $G$ and let $G_f$ be defined as in \cref{thm:minstcut_representation}.
    By~\cref{lem:minstcut:forcing}, it suffices to find a smallest edge set $F \subseteq \cutset_G(X)$ by which all the free vertices in $G_f$ are dominated.

    To this end, we construct a bipartite multigraph $H$ with vertex set $X_H \cup \overline{X}_H$ as follows.
    Let $\mathcal C_X$ and $\mathcal C_{\overline{X}}$ be the sets of strongly connected components in $G_f[X]$ and in $G_f[V \setminus X]$, respectively.
    The vertex set of $H$ is defined as $X_H = \{v_{C} : C \in \mathcal C_X\}$ and $\overline{X}_H = \{w_C : C \in \mathcal C_{\overline{X}}\}$.
    For each edge $xy \in \cutset_G(X)$ with $x \in X$ and $y \notin X$, we add an edge between $v_C$ and $w_D$, where $C \in \mathcal C_X$ and $D \in \mathcal C_{\overline{X}}$ are the strongly connected components that $x$ and $y$ belong to, respectively.
    Since $H$ is a multigraph, the set of edges in $H$ bijectively corresponds to $\cutset_G(X)$.
    Recall that $C_s$ and $C_t$ are the strongly connected components in $G_f$ that contain $s$ and $t$, respectively.
    Let $Y_H$ consist of the vertices $v_C$ such that $C \neq C_s$ and no arc of $G_f[X]$ is outgoing from $C$ to another strongly connected component in $\mathcal C_{X}$. 
    Similarly, let $\overline Y_H$ consist of the vertices $w_C$ such that $C\neq C_t$ and no arc of $G_f[V\setminus X]$ is incoming to $C$ from another strongly connected component in $\mathcal C_{\overline{X}}$.
    
    For a vertex set $T$ of a graph, a \emph{$T$-edge cover} is a set $F$ of edges such that for every vertex in $T$, there is an edge incident to it in $F$.
    A minimum-cardinality $T$-edge cover of a graph can be computed in $\bigoh(nm)$ time by a standard maximum matching algorithm~\cite{lawler1976combinatorial}.
    We now claim that $F$ is a forcing set for $\Lambda_{st}(G)$ with $F \subseteq \cutset_G(X)$ if and only if $F$ is a $(Y_H \cup \overline{Y}_H)$-edge cover of $H$.

    Suppose that $F$ is a forcing set for $\Lambda_{st}(G)$.
    By~\cref{lem:minstcut:forcing}, $F$ dominates all the free vertices in $G_f$.
    In particular, $F$ dominates all the vertices in each component in $\mathcal C_X \setminus \{C_s\}$ and in $\mathcal C_{\overline{X}} \setminus \{C_t\}$.
    Let $C\in \mathcal{C}_X$ with $v_C\in Y_H$, and let $uv \in F$ where $u\in X$ and $v\in V\setminus X$, be an edge that dominates a vertex of $C$.
    Then that vertex can reach $u$ in $G_f$.
    Since $C$ has no outgoing arc to another strongly connected component in $\mathcal{C}_X$, we have $u\in C$.
    Hence, $F$ contains an edge corresponding to an edge of $H$ incident with $v_C$.
    The analogous statement holds for $w_C\in\overline{Y}_H$.
    Thus, $F$ is a $(Y_H \cup \overline{Y}_H)$-edge cover of $H$.
    % This implies that for each $v_C \in Y_H \cup \overline{Y}_H$, $S$ contains at least one edge incident to it, meaning that $S$ is a $(Y_H \cup \overline{Y}_H)$-edge cover of $H$.

    Conversely, suppose that $H$ has a $(Y_H\cup\overline Y_H)$-edge cover $F$.
    Let $w$ be a free vertex in $X$.
    Let $C$ be a strongly connected component that is reachable from $w$ and has no outgoing arc in $G_f[X]$.
    Then, we have $C \neq C_s$ due to \cref{obs:unique-source-sink}.
    % Let $ab \in F$ an edge with $a \in X_H$ and $b \in \overline{X}_H$ that covers $v_C$, namely $v_C = a$.
    % By~\cref{obs:unique-source-sink}, there are paths from $s$ to $w$ and from $b$ to $t$.
    % By concatenating these two paths with a path from $w$ to $C$, we conclude that $w$ is dominated by $(a, b)$.
    Let $uv$ be the edge of $G$ corresponding to an edge of $F$ incident with $v_C$, where $u\in C$ and $v\notin X$.
    Since $C$ is reachable from $w$ and strongly connected, $w$ can reach $u$.
    By~\cref{obs:unique-source-sink}, there is an $s$--$t$ walk passing through both $w$ and $(u,v)$, so $uv$ dominates $w$.
    Symmetrically, every free vertex in $V \setminus X$ is dominated by an edge in $F$.
    Hence, by~\cref{lem:minstcut:forcing}, $F$ is a forcing set.
\end{proof}
\antiForcingSpecifiedMinSTCut*
\begin{proof}
    Let $G = (V, E, \weight)$ be an input graph and let $\{X, V\setminus X\}$ be a given minimum $s$--$t$ cut of $G$ with $s \in X$.
    We can assume that $G$ is connected.
    By~\cref{apd:lem:minstcut:anti-forcing}, it suffices to find a smallest edge set $F \subseteq E \setminus \cutset_G(X)$ that hits the cutset $\cutset_G(X')$ of each minimum $s$--$t$ cut $\{X', V\setminus X'\}$ satisfying either $X' \subsetneq X$ or $X \subsetneq X'$.
    Let $\mathcal S_X$ (resp.~$\mathcal S_{\overline{X}}$) be the set of all cutsets $\cutset_{G}(X') \in \Lambda_{st}(G)$ whose $s$-shore $X'$ satisfies $X' \subsetneq X$ (resp.~$X \subsetneq X'$).
    Since edges in $G[V\setminus X]$ (resp.~$G[X]$) do not hit any cut in $\mathcal S_{X}$ (resp.~$\mathcal S_{\overline{X}}$), we can independently compute a smallest hitting set for each of $\mathcal S_X$ and $\mathcal S_{\overline{X}}$.
    Thus, by considering symmetry, it suffices to compute a smallest set $F \subseteq E(G[X])$ that hits $\mathcal S_X$.
    This problem is known to be solvable in polynomial time~\cite{Frank81,NalamS25}.
    To be self-contained, we give an explicit reduction from this problem to that of finding a minimum-cost arborescence in an arc-weighted directed graph.

    Let $D$ be a directed graph on vertex set $X$. 
    For every arc $(u, v)$ of $G_f[X]$, add the opposite arc $(v, u)$ to $D$ with weight~0. 
    Moreover, if $(u,v)$ is relevant, add the arc $(u, v)$ to $D$ with weight~1.
    The sets of arcs of weight~1 and those of weight~0 are denoted by $A_1$ and $A_0$, respectively.
    We claim that the minimum weight of an (out-)arborescence rooted at $s$ in $D$ is equal to the smallest cardinality of a set $F\subseteq E(G[X])$ that intersects every cutset in $\mathcal{S}_X$.
    % hitting set of $\mathcal S_{X}$.
    
    Let $F \subseteq E(G[X])$ be a hitting set of $\mathcal{S}_{X}$ with $|F| \le k$.
    We construct a spanning subgraph $D'$ of $D$ by setting $E(D') = \{(u, v) \in A_1: uv \in F\} \cup A_0$.
    Since the total weight of $D'$ is at most $k$, it suffices to show that $D'$ has an arborescence rooted at $s$, that is, every vertex in $X$ is reachable from $s$ in $D'$.
    % Suppose that there is a vertex $v$ that is not reachable from $s$ in $D'$.
    % Let $Y \subseteq X$ be a vertex set such that $v \in Y$ and $s \notin Y$ and there are no arcs from $X \setminus Y$ to $Y$ in $D'$.
    Suppose that there is a vertex in $X$ that is not reachable from $s$ in $D'$, and let $Y$ be the set of vertices not reachable from $s$.
    For each arc in $D'$ with weight~$1$, $D'$ has an opposite arc with weight~$0$.
    Thus, every arc from $Y$ to $X \setminus Y$ is in $A_0$.
    This implies that $X \setminus Y$ is closed in $G_f$.
    Since $F$ hits $\cutset_G(X\setminus Y)\in\mathcal{S}_X$ and $F \subseteq E(G[X])$, some edge of $F$ joins $X\setminus Y$ and $Y$, whose corresponding arc belongs to $D'$, a contradiction.

    Conversely, suppose that $D$ has an arborescence $T$ rooted at $s$ with weight at most~$k$.
    We define $F = \{uv \in E(G[X]) : (u, v) \in E(T) \cap A_1\}$ and claim that $F$ hits all the cutsets in $\mathcal S_X$.
    Let $\cutset_{G}(X') \in \mathcal S_X$ with $s$-shore $X' \subsetneq X$.
    Let $v \in X \setminus X'$ and let $P$ be a directed path from $s$ to $v$ in $T$.
    Let $(u, w)$ be the first arc of $P$ leaving $X'$.
    Since $X'$ is closed in $G_f$, $(u, w)$ cannot belong to $A_0$.
    Therefore, $(u, w)\in A_1$, and hence $uw\in F$, which hits $\cutset_{G}(X')$.

    The directed graph $D$ can be computed in $\bigoh(nm)$ time.
    It is well known that a minimum-weight arborescence in an arc-weighted directed graph can be computed in $\bigoh(m + n\log n)$ time~\cite{GabowGST86}.
    This completes the proof of the theorem.
\end{proof}

\end{document}